\documentclass[11pt]{article}

\usepackage{amsfonts,amssymb,amsmath}
\usepackage{multirow,enumitem,xcolor,ae,aecompl} 
\usepackage{graphicx}

\newcommand{\FF}{\vspace*{\medskipamount}}

\newcommand{\BBB}{\vspace*{-\bigskipamount}}

\newcommand{\cA}{\mathcal{A}}

\newcommand{\cC}{\mathcal{C}}

\newcommand{\cE}{\mathcal{E}}

\newcommand{\cG}{\mathcal{G}}

\newcommand{\cI}{\mathcal{I}}

\newcommand{\cO}{\mathcal{O}}

\newcommand{\mN}{\mathbb{N}}

\newlength{\figurewidth}
\newcommand{\remove}[1]{}

\newlength{\captionwidth}
\newcommand{\Paragraph}[1]{\BBB\paragraph{#1}}

\newcommand{\qed}{\hfill $\square$ \smallbreak}

\newenvironment{proof}{\noindent\textbf{Proof:}}{\qed}

\newtheorem{theorem}{Theorem}
\newtheorem{lemma}{Lemma}

\newtheorem{corollary}{Corollary}
\newtheorem{proposition}{Proposition}

\begin{document}

\renewcommand{\baselinestretch}{1.1}

\parskip 1pt

\title{Disconnected Agreement in Networks Prone to Link Failures
		\vfill}

\author{Bogdan S. Chlebus \footnotemark[1]
	\and
	Dariusz R. Kowalski \footnotemark[1] 
	\and
	Jan Olkowski \footnotemark[2]
	\and
	J\k{e}drzej Olkowski \footnotemark[3]}

\footnotetext[1]{School of Computer and Cyber Sciences, Augusta University, Augusta, Georgia, USA.}

\footnotetext[2]{Department of Computer Science, 	University of Maryland, College Park, Maryland, USA.}

\footnotetext[3]{Wydzia{\l} Matematyki, Mechaniki i Informatyki, Uniwersytet Warszawski, Warszawa, Poland.}

\date{}

\maketitle

\vfill

\begin{abstract}
We consider deterministic distributed algorithms for reaching agreement in synchronous networks of arbitrary topologies.
Links are bi-directional and prone to failures while nodes stay non-faulty at all times.
A faulty link may omit messages.
Agreement among nodes is understood as holding in each connected component of a network obtained by removing faulty links.
We call ``disconnected agreement'' the algorithmic problem of reaching such agreement.
We introduce the concept of stretch,  which is the number of connected components of a network, obtained by removing faulty links, minus~$1$ plus the sum of diameters of connected components.
We define the concepts of ``fast'' and ``early-stopping''  algorithms for disconnected agreement by referring to stretch.
We consider trade-offs between the knowledge of nodes, size of messages, and running times of algorithms. 
A network has $n$ nodes and  $m$ links.
Nodes are normally assumed to know their own names and ability to associate communication with local ports. 
If we additionally assume that a bound~$\Lambda$ on stretch is known to all nodes, then there is an algorithm for disconnected agreement working in time $\cO(\Lambda)$ using messages of $\cO(\log n)$ bits.
We give a general disconnected agreement algorithm operating in~$n+1$ rounds that  uses messages of $\cO(\log n)$ bits.
Let~$\lambda$ be an unknown stretch occurring in an execution; we give an algorithm working in time~$(\lambda+2)^3$ and using messages of $\cO(n\log n)$ bits.
We show that disconnected agreement can be solved in the optimal $\cO(\lambda)$ time, but at the cost of increasing message size to~$\cO(m\log n)$. 
We also design an algorithm that uses only~$\cO(n)$ non-faulty links and works in time~$\cO(n m)$, while nodes start with their ports mapped to neighbors and messages carry $\cO(m\log n)$ bits.
We prove lower bounds on the performance of disconnected-agreement solutions that refer to the parameters of  evolving network topologies and the knowledge available to nodes.

\vfill
\end{abstract}

\vfill

~

\thispagestyle{empty}

\setcounter{page}{0}

\newpage

\section{Introduction}

\label{sec:introduction}

We introduce a variant of agreement and present deterministic distributed algorithms for this problem in synchronous networks.
Nodes represent processing units and links model bi-directional  communication channels between pairs of nodes.
Links are prone to failures but nodes stay operational at all times.
A faulty link may not  convey a  message transmitted at a round. 

Once a link omits a message, it may omit messages in the future.
A link that has omitted a message manifested its faultiness and is considered \emph{unreliable}.
A safe approach for a communication algorithm is to treat unreliable links as if they crashed and not use them for transmissions after they manifested unreliability. 
Based on this intuition, we model a network with link failures as evolving through a chain of sub-networks obtained by removing unreliable links.

The disconnected agreement problem  is about nodes reaching  agreement on a common value.
Each node starts with an input value and it eventually decides on some value.
The requirements of reaching disconnected agreement are modeled after the problem of consensus.
That problem is defined by the three requirements of termination, validity, and agreement, which constrain the output values and the process of deciding.
Those conditions may depend on the nature of faults in a distributed system, for example there are different conditions specifying what the consensus problem is in the models of nodes prone to crashes and for arbitrary (Byzantine) faults of nodes.
We address a problem like consensus in a model of dynamic networks whose topologies evolve because of removal of faulty links.

Deleting links affects the topology of a network, possibly even decomposing a network into disjoint connected components.
If such a decomposition occurs fast enough, it may be impossible for some nodes to learn of input values in other components whose presence affects the ultimate decision for a given agreement-seeking algorithm.
To make reaching agreement possible, we need either to restrict link failures, or to modify the specification of the agreement problem, or both.

In this work we study a scenario in which a faulty link may not transmit a message through at some rounds. 
Such an omission fault is a relatively benign type of link failure.
Messages that are delivered, even through faulty links, are not tempered with and are received as transmitted.
We do not impose any restrictions on which links may be faulty, in that, a removal of faulty links may produce an arbitrary sub-network.
In particular, it is possible that some links simply crash such that this disconnects the network.

The problem to reach agreement needs to be suitably reformulated to reflect the possibility of a network becoming decomposed into disconnected sub-networks by removal of faulty links.
We study agreement that allows nodes in different connected components of the network obtained by removing unreliable links to decide on different values but still requires nodes within a connected component to decide on the same  value. 
It is natural not to demand that nodes in different connected components decide on the same value, due to impossibility to accomplish such a goal without a proper flow of information.
The process of breaking into connected components occurs in time, which creates a challenge for the nodes in a connected component to decide on the same value. 
Any two nodes may eventually end up in different connected components but what matters are connected components as they exist at the rounds in which such two nodes actually decide.

Our approach to specifying disconnected agreement subsumes the models of link crashes, which is stricter in that once a link manifests its faultiness by omitting a message it will omit all the messages in the future.
Node crashes may also decompose a network into separate connected components, and our approach subsumes the model of node crashes as well.
To see this for node crashes, observe that a node's crash could be simulated by crashing all the links incident to a node at the round of its crash.
If a node's crash occurs in such a simulation then the node stays operational but it constitutes a single-node connected component, so any decision by the node on some initial input value  satisfies agreement and validity.  

Let $n$ denote the number of nodes and $m$ the number of links in an initial network.
Nodes are equipped with names, which are unique integers represented by $\cO(\log n)$ bits.
Initial input values are assumed to be encodable with $\cO(\log n)$ bits.
A message is considered ``short'' if it carries $\cO(\log n)$ bits, and it is ``linear'' if it carries $\cO(n\log n)$ bits.
A network operates with ``minimal knowledge'' if each node initially  knows its individual name, the initial input value,  and it can distinguish among ports as sources of the incoming and outlets for the outgoing communication.

We use a network's dynamic attribute called ``stretch,'' which is an integer determined by the number of connected components and their diameters.
The purpose of using stretch is to consider scalability of disconnected agreement solutions to networks evolving through link failures and their removal from the network. 
In particular, we define ``fast''  and ``early-stopping''  algorithms solving disconnected agreement by referring to stretch.
Along with time performance, we consider the amount of communication, measured as an upper bound on the number of bits in each individual message.
We also study performance of algorithms measured by ``link use'' defined as the number of reliable links through which nodes transmit messages.

\Paragraph{A summary of the results.}

We introduce the problem of disconnected agreement and give deterministic algorithms for this problem in synchronous networks with links prone to failures.
Faulty links may omit messages.
An upper bound on stretch, denoted $\Lambda$, could be given to all nodes, with an understanding that faults occurring in an execution are restricted such that the actual stretch never surpasses~$\Lambda$.
An algorithm solving disconnected agreement with a known upper bound~$\Lambda$ on the stretch  is considered ``fast'' if it runs in time~$\cO(\Lambda)$.
A fast solution to disconnected agreement is discussed in Section~\ref{sec:fast-agreement}. 
We also show a lower bound which demonstrates that, for each natural number $\lambda$ and an algorithm solving disconnected agreement in networks prone to link failures, there exists a network that has stretch $\lambda$ and such that each execution of the algorithm on this network takes at least $\lambda$ rounds. 
In Section \ref{sec:minimal-knowledge}, we show how to solve disconnected agreement in~$n+1$ rounds with short messages in networks where nodes have minimal knowledge.
We give an algorithm relying on minimal knowledge working in time $(\lambda+2)^3$ and using linear messages of $\cO(n\log n)$ bits, where $\lambda$ is an unknown stretch  occurring in an execution; this algorithm is presented in Section~\ref{sec:linear-messages}.
A disconnected agreement solution is considered ``early-stopping'' if it operates in time proportional to the unknown stretch actually occurring in an execution.
In Section~\ref{sec:early-stopping}, we develop an early-stopping solution to disconnected agreement relying on minimal knowledge that employs messages of $\cO(m\log n)$ bits. 
We propose to count the number of reliable links used by a communication algorithm as its performance metric.
To make this performance measure meaningful, the nodes need to start knowing their neighbors, in  having a correct mapping of communication ports to neighbors.
In Section~\ref{sec:link-use}, we give a solution to disconnected agreement that uses at most an asymptotically optimum number $2n$  of reliable links and works in $\cO(n m)$ rounds, without knowing the size $n$ of the network.
We also show, in Section~\ref{sec:link-use}, that if the nodes start with their ports not mapped on neighbors,  then any disconnected agreement solution has to use $\Omega(m)$ links in some networks of $\Theta(m)$ links, for all numbers~$n$ and~$m$ such that $n\le m\le n^2$.
A summary of algorithms developed in this paper, with their performance bounds, is  given in Table~\ref{table:summary}.


\newcommand{\RB}{\raisebox{12pt}{}}
\newcommand{\LB}{\raisebox{-6pt}{}}

\begin{table}[t]
\begin{center}
\begin{tabular}{c || l | l | l | l | l  }
algorithm / section& time & message size & \# links & knowledge & lower bound  
\LB\\ 
\hline\hline
\RB \LB
\textsc{Fast-Agreement}\,$(\Lambda)$  / \ref{sec:fast-agreement}& $\Lambda$ \hfill $^{~~\dag}$  &  $\cO(\log n)$ & $\cO(m)$ &  $\Lambda$ known  & time $\lambda\le\Lambda$  \\ 
\hline 
\RB \LB
\textsc{SM-Agreement} / \ref{sec:minimal-knowledge}&  $n+1$ &  $\cO(\log n)$ &  $\cO(m)$ &  minimal & time $\lambda$   \\ 
\hline
\RB \LB
\textsc{LM-Agreement} / \ref{sec:linear-messages} & $(\lambda+2)^3$ & $\cO(n\log n)$&$\cO(m)$ & minimal  & time $\lambda$  \\
\hline
\RB \LB
\textsc{ES-Agreement} / \ref{sec:early-stopping} & $\lambda+2$ \hfill $^{~\dag}$ & $\cO(m\log n)$ & $\cO(m)$ & minimal  & time $\lambda$ \\ 
\hline
\RB \LB
\textsc{OL-Agreement} / \ref{sec:link-use} & $\cO(n m)$ & $\cO(m\log n)$ & $2n$ \hfill $^{~~\dag}$ & neighbors known  & \# links $\Omega(n)$  
\end{tabular}
\end{center}

\caption{\label{table:summary}
A summary of the given deterministic distributed  algorithms for disconnected agreement and their respective performance bounds.
The notation~$n$ means the number of nodes and $m$ the number of links in the initial network; none of these numbers is ever assumed to be known.
The parameter~$\Lambda$ denotes an upper bound on stretch if it is known.
The notation $\lambda$ means the stretch actually occurring in an execution by the time all nodes halt.
The dagger symbol $\dag$ indicates the asymptotic optimality of the respective upper bound.
}
\end{table}

\Paragraph{The previous work on agreement in networks.}

Dolev~\cite{Dolev82} studied Byzantine consensus in networks with faulty nodes and gave connectivity conditions sufficient and necessary for a solution to exist; see also Fischer et al.~\cite{FischerLM86}, and Hadzilacos~\cite{Hadzilacos87}.
Khan et al.~\cite{KhanNV19} considered a related problem in the model with restricted Byzantine faults, in particular, in the model requiring a node to broadcast identical messages to all neighbors at a round.
Tseng and Vaidya~\cite{TsengV15}  presented necessary and sufficient conditions for the solvability of consensus  in directed graphs under the models of crash and Byzantine failures.
In a related work, Choudhury et al.~\cite{ChoudhuryGPRS18} provided a time-optimal algorithm to solve consensus in directed graphs with node crashes.
Casta{\~{n}}eda et al.~\cite{CastanedaFPRRT19} considered networks with nodes prone to crashes with an upper bound~$t$ on the number of crashes and showed that, as long at the network remained $(t+1)$-connected, Consensus was solvable in a number of rounds determined by how conducive the network was to broadcasting.
Chlebus et al.~\cite{ChlebusKO20} studied networks with Byzantine nodes such that the removal of faulty nodes leaves a network that is sufficiently connected; they gave fast solutions to consensus and showed a separation of consensus with Byzantine nodes from consensus with Byzantine nodes using message authentication, with respect to asymptotic time performance in suitably connected networks.
Tseng~\cite{Tseng16} and Tseng and Vaidya~\cite{TsengV16} surveyed work on reaching agreement in networks subject to node faults.
Winkler and Schmidt~\cite{WinklerS2019} gave a survey of recent work on consensus in directed dynamic networks.


\begin{table}[t]
\begin{center}
\begin{tabular}{c || l | l | l | l | l  }
algorithm & time & message size & \# links & knowledge & comments
\LB\\ 
\hline\hline
\RB \LB
Kuhn et. al.~\cite{KuhnLO10} & $O(n)$  &  $\cO(\log n)$ & $\cO(m)$ &  minimal &  for $T = \Omega(n)$  \\ 
\hline 
\RB \LB
Biely et. al~\cite{BielyRSSW18}, alg. $2$ & $2D + 2E$ & $O(m\log{n})$ & $\cO(m)$ & $D$ known & directed links  \\ 
\hline
\RB \LB
Biely et. al~\cite{BielyRSSW18}, alg. $4$ & $3D + 3$ & $\cO(nD\log n)$ & $\cO(m)$ & $D$ known & k-set agreement \\
\hline
\RB \LB
Kuhn et. al~\cite{KuhnOM11} & $O(n)$ & $\cO(m^2\log n)$ & $\cO(m)$ & minimal & $\Delta$-coordinated  
\end{tabular}
\end{center}
\caption{\label{table:previous-work}
A summary of the previous consensus results that are most relevant to the model studied in this paper and their respective performance bounds.
Notation~$n$ denotes the number of nodes and $m$ the number of links in an initial network; none of these numbers are ever assumed to be known. 
Letter~$T$ denotes a number such that in every interval of $T$ rounds there exists a stable connected spanning subgraph, see~\cite{KuhnLO10}. 
Letter~$D$ denotes a dynamic source diameter, and $E$ denotes a dynamic graph depth, see~\cite{BielyRSSW18}. 
The stretch  and parameters $D$ and $E$ in bidirectional networks with unreliable links are related as follows: $D = E = \Lambda \ge \lambda$.
Parameter $\Delta$ is an upper bound on the difference in the respective times of termination for any two nodes, see~\cite{KuhnOM11}.
}
\end{table}

The approach to failing links that we work with falls within the scope of modeling dynamic evolution of networks.
Next, we discuss previous work on solving consensus in networks undergoing topology changes, malfunctioning links and transmission failures. 
Perry and Toueg~\cite{PerryT86}, Santoro and Widmayer~\cite{SantoroW89,SantoroW07},  and Charron{-}Bost and Schiper~\cite{Charron-BostS09} studied agreement problems in complete networks in the presence of dynamic transmission failures.
Kuhn et al.~\cite{KuhnOM11} considered $\Delta$-coordinated binary consensus in undirected graphs, whose topology could change arbitrarily from round to round, as long it stayed connected; here $\Delta$ is a parameter that bounds from above the difference in times of termination for any two nodes. 
Paper~\cite{KuhnOM11} showed how to solve $\Delta$-coordinated binary consensus in  $\cO(\frac{nD}{D+\Delta} + \Delta)$ rounds using message of $\cO(m^2\log{n})$ size  without a prior knowledge of the network's diameter~$D$. 
Augustine et al.~\cite{AugustinePRU12} studied dynamic networks with adversarial churn and gave randomized algorithm to reach almost-everywhere agreement with high probability  in poly-log\-a\-rith\-mic time.
Charron-Bost et al.~\cite{Charron-BostFN15} considered approximate consensus in dynamic networks and provided connectivity restrictions on network evolution  that make approximate consensus solvable. 
Coulouma et al.~\cite{CouloumaGP15} characterized oblivious message adversaries for which consensus is solvable.
Biely et al.~\cite{BielyRSSW18} considered reaching agreement and $k$-set agreement in networks when communication is modeled by  directed-graph topologies controlled by adversaries, with the goal to identify constraints on adversaries to make the considered problems solvable.
Paper~\cite{BielyRSSW18} solved $k$-set agreement in time $\cO(3D+H)$ and using messages of $\cO(nD\log{n})$ size, where $D$ denotes the dynamic source diameter and $H$ denotes the dynamic graph depth, and the code of algorithm includes~$D$. 
Kuhn et al.~\cite{KuhnLO10} considered dynamic networks in which the network topology changes from round to round such that  in every $T\ge 1$ consecutive rounds there exists a stable connected spanning subgraph, where $T$ is a parameter. 
Paper~ \cite{KuhnLO10} gave an algorithm that implements any computable function of the initial inputs, working in $\cO(n + n^2 / T)$ time with messages of $\cO(\log{n} + d)$ size, where $d$ denotes the size of input values.
Biely et al.~\cite{BielySW11} studied consensus in a synchronous model combining transient process and communication failures, with the numbers of such failures explicitly bounded, and demonstrated adaptability of popular consensus solutions to this model.
Schmid et al.~\cite{SchmidWK09} showed impossibility results and lower bounds for the number of processes and rounds for synchronous agreement under transient link failures.
Winkler et al.~\cite{WinklerSS19} investigated solving consensus in dynamic networks with links controlled by adversaries who  only eventually provide a desired behavior of networks.
A representative selection of previous results related to the model of this paper is given in Table~\ref{table:previous-work}.

Next, we briefly discuss advances on other algorithmic problems in dynamic networks.
Cornejo et al.~\cite{CornejoGN12} considered the aggregation problem in dynamic networks, where the goal is to collect  tokens, originally distributed among the nodes, at a minimum number of nodes.
Haeupler and Kuhn~\cite{HaeuplerK12} developed lower bounds on information dissemination in adversarial dynamic networks.
Sarma et al.~\cite{SarmaMP12} investigated algorithms interpreted as random walks in dynamic networks.
Michail et al.~\cite{MichailCS13} studied naming and counting in anonymous dynamic networks.
Michail et al.~\cite{MichailCS14} studied propagation of influence in dynamic networks  that may be disconnected at any round.
Frameworks and models of distributed computing in networks evolving in time were discussed in~\cite{CasteigtsFQS12, KuhnLO10,KuhnO11, Michail16}.

\Paragraph{The previous work on the scalability of consensus solutions.}

We review the previous  work on consensus solutions that scale well with their performance metrics, including communication and running time.
Let the letter~$t$ denote an upper bound on the number of node failures that is known to all nodes, and the letter~$f$ be an actual number of failures occurring in an execution, which is not assumed to be known.
We may call a consensus solution ``fast'' if it operates in time $\cO(t+1)$ and ``early stopping'' if its running time is $\cO(f+1)$.
Scaling communication performance can be understood either ``globally,'' which conservatively means $\cO(n\text{ polylog } n)$ total communication, or ``locally,'' which conservatively means $\cO(\text{polylog }n)$ communication generated by each communicating agent.

Networks of degrees uniformly bounded by a constant provide an ultimate local scaling of communication.
Upfal~\cite{Upfal94} showed that an almost-everywhere agreement can be solved in networks of constant degree with a linear number of crashes allowed to occur.
Dolev et al.~\cite{DolevRS90} gave an early stopping solution for consensus with arbitrary process failures and a lower bound $\min\{ t+1,f+2\}$ on the number of rounds. 
Berman et al.~\cite{BermanGP92} developed an early-stopping consensus solution with arbitrary process failures that was simultaneously optimal with respect to time performance $\min\{ t+1,f+2\}$ and the number of processes $n>3t$.
Galil et al.~\cite{GalilMY95} developed a crash-resilient consensus algorithm using $\cO(n)$ messages, thus showing that this number of messages is optimal, but their algorithm runs in over-linear time $\cO(n^{1+\varepsilon})$, for any $0<\varepsilon<1$; they also gave an early-stopping algorithm with $\cO(n+fn^\varepsilon)$ message complexity, for any $0<\varepsilon<1$.
Chlebus et al.~\cite{ChlebusKO-PODC23} presented a binary consensus algorithm operating in $\cO(t+\log n)$ rounds and sending $\cO(n + t\log t)$ messages for $t<\frac{n}{5}$; the algorithm sends the optimum
number of bits/messages $\cO(n)$ in the single-port model, as long as $t = \cO( n/\log n )$. 
Garay and Moses~\cite{GarayM98} presented a fast consensus solution for arbitrary processor faults that runs on the optimum $n>3t$ number of processors while using the amount of communication that is polynomial in~$n$; they also gave an early-stopping variant of this algorithm. 
Chlebus and Kowalski~\cite{ChlebusK-JCSS06} developed a gossiping algorithm coping with crashes and applied it to develop a consensus solution that is fast, by running in $\cO(t+1)$ time, while sending $\cO(n \log^2 t)$ messages, provided that $n-t=\Omega(n)$.
Chlebus and Kowalski \cite{ChlebusK-DISC06} developed a deterministic early-stopping  algorithm that scales communication globally by sending $\cO(n \log^5 n)$ messages.
Coan~\cite{Coan93} gave an early-stopping consensus solution that uses messages of size logarithmic in the range of input values; see also Bar-Noy et al.~\cite{Bar-NoyDDS92} and Berman et al.~\cite{BermanGP92}.
Chlebus at al.~\cite{ChlebusKS-PODC09} gave a fast deterministic  algorithm for consensus which has nodes send $\cO(n\log^4 n)$ bits, and showed that no deterministic  algorithm can be locally scalable with respect to message complexity.
Chlebus and Kowalski~\cite{ChlebusK-SPAA09} gave a randomized consensus solution terminating in the expected $\cO(\log n)$ time, while  the expected number of bits that each process sends and receives against oblivious adversaries is~$\cO(\log n)$, assuming that a bound~$t$ on the number of crashes is a constant fraction of the number~$n$ of nodes.  
Chlebus et al.~\cite{ChlebusKS-DISC10} gave  a scalable quantum algorithm to solve binary consensus, in a system of $n$  crash-prone quantum processes, which works in $\cO(\text{polylog }n)$ time  sending $\cO(n \text{ polylog } n)$ qubits against the adaptive adversary.
Dolev and Lenzen~\cite{DolevL13} showed that any crash-resilient consensus algorithm deciding in $f + 1$ rounds has worst-case message complexity $\Omega(n^2f)$.
Alistarh et al.~\cite{AlistarhAKS18} developed a randomized algorithm for asynchronous message passing that scales well to the number of crashes, in that it sends the expected $\cO(n t + t^2 \log^2 t)$ messages.

\Paragraph{A general perspective.}

Algorithmic problems concerning reaching agreement are central to the study of distributed systems and communication networks.
The probem of consensus was introduced by Pease et al.~\cite{PeaseSL80} and Lamport et al.~\cite{LamportSP82}. 
Books by Attiya and Welch~\cite{Attiya-Welch-book2004}, Herlihy and Shavit~\cite{HerlihyShavit-book}, Lynch~\cite{Lynch-book96}, and Raynal~\cite{Raynal2010-Synthesis} provide general expositions of formal models of distributed systems as frameworks to develop distributed algorithms in systems prone to failures of components. 
Attiya and Ellen~\cite{Attiya-Ellen-book-2014} and Herlihy et al.~\cite{HerlihyKozlovRajsbaum-book} present techniques to show  impossibilities and lower bounds for problems in distributed computing.

\section{Preliminaries}

\label{sec:preliminaries}

We model distributed systems as collections of  nodes that communicate through a wired communication network.
Executions of distributed algorithms are synchronous, in that they are partitioned into global rounds coordinated across the whole network.
There are $n$ nodes in a network. 
Each node has a unique \emph{name} used to determine its identity; a name can be encoded by  $\cO(\log n)$ bits.
Nodes start with a read-only variable \texttt{name} initialized to their names; the private copy of this variable at a node~$p$ is denoted by \texttt{name}$_p$.
Every node identified other nodes by their names.

A node has distinctly labeled ports that act as interfaces to links connecting the node to its neighbors.
We say that a node \emph{knows neighbors} if it can associate with each port the name of a node that receives communication sent through this port and whose communication is received through this port.
Links connecting pairs of nodes serve as bi-directional  communication channels.
Messages are scaled to channel capacity, in that precisely one message can be transmitted at a round through a link in each direction. 
The size of a message denotes the number of bits used to encode it.
A message transmitted at a round is delivered within the same round.
If at least one message is transmitted by a link in an execution then this link is \emph{used} and otherwise it is \emph{unused} in this execution.

A link may fail to deliver a message transmitted through it at a round; once such omission happens for a link, it is considered \emph{unreliable}.
The functionality of an unreliable  link is unpredictable, in that it may either deliver a transmitted message or fail to do it.
A link that has never failed to deliver a message by a given round is \emph{reliable} at this round.
A path in the network is \emph{reliable} at a round if it consists only of links that are reliable at this round. 

Nodes and links of a network can be interpreted as a simple graph, with nodes serving as vertices and links as undirected edges.
A network at the start of an execution is represented by some \emph{initial graph~$G$}, which is simple and connected.
An edge representing an unreliable link is removed from the  graph $G$ at the first round it fails to deliver a transmitted message. 
A graph representing the network evolves through a sequence of its sub-graphs and may become partitioned into multiple connected components.
Once an algorithm's execution halts, we stop this evolution of the initial graph~$G$.
An evolving  network, and its graph representation~$G$, at the first round after all the nodes have halted in an execution is called \emph{final} for this execution and denoted by~$G_F$.

Nodes start an execution of a distributed algorithm simultaneously.
A node performes the following actions at a round: it sends messages through some of its ports,  collects messages sent by neighbors, and performs local computation.
A node can send messages to any subset of its neighbors and collect messages coming from all neighbors  at a round.

\Paragraph{Disconnected agreement.}

We precisely define the algorithmic problem of interest as follows.
Each node~$p$ starts with an initial value \texttt{input}$_p$.
We assume two properties of such input values.
One is that an input value can be represented by $\cO(\log n)$ bits.
The other is that input values can be compared, in the sense of belonging to a domain with a total order.
In particular, finitely many initial input values contain a maximum one. 
We say that a node \emph{decides} when it produces an output by seting a dedicated variable to a decision value.
The operation of deciding is irrevocable.
An algorithm \emph{solves disconnected agreement} in networks with links prone to failures if the following three properties hold in all executions:
\begin{description}
\item[\sf Termination:] every node eventually decides.

\item[\sf Validity:] each decision value is among the input values.

\item[\sf Agreement:]
when a node~$p$ decides then its decision value is the same as these of the nodes that have already decided and to which $p$ is connected by a reliable path at the round  of deciding.
\end{description}
An equivalent formulation of the agreement property is to say that the decisions of all the nodes in a connected component of $G_F$ are equal. 
This is because if a node~$p$ decides at a round and $q$ is a node that has already decided to which $p$ is still connected by a reliable path at this round then if that path stays reliable until halting then $p$ and $q$ end up in the same connected component of~$G_F$.

\Paragraph{The size of messages.}

If a message sent by a node executing a disconnected agreement solution carries a constant number of node names and a constant number of  input values then the size of such a message is $\cO(\log n)$ bits, due to our assumptions about encoding names and input values.
Messages of $\cO(\log n)$ bits are called \emph{short}.
If a message carries $\cO(n)$ node names and $\cO(n)$ input values then the size of such a message is $\cO(n \log n)$ bits.
We call  messages of $\cO(n \log n)$ bits \emph{linear}.

\Paragraph{Stretch.}

Let $H$ be a simple graph.
If $H$ is connected then $\text{diam}(H)$ denotes the diameter of~$H$.
Suppose $H$ has $k$ connected components $C_1,\ldots, C_k$, where $k\ge 1$, and let $d_i=\text{diam}(C_i)$ be the diameter of component~$C_i$.
The \emph{stretch of~$H$} is defined as a number $k-1+\sum_{i=1}^k d_i$.
The stretch of a connected graph equals its diameter, because then $k=1$.
The stretch of $H$ can be interpreted as the maximum diameter of a graph obtained from $H$ by adding $k-1$ edges such that the obtained graph is connected.
The maximum stretch of a graph with $n$ vertices is $n-1$, which occurs when every vertex is isolated or, more generally, when each connected component is a line of nodes.

\Paragraph{Knowledge.}

A property of distributed communication environments or executions is \emph{known} if it can be used in codes of algorithms. 
We say that an algorithm relies on \emph{minimal knowledge} if each node knows  its unique name and can identify a port through which a message arrives and can assign a port for a message to be transmitted through.
The number of nodes in a network $n$ is never assumed to be known in this paper.

Neighbors can be discovered at one round of communication by all nodes sending their names to the neighbors: incoming messages allow to assign the sender's name to a port.
This operation requires transmitting through every link in the network.
If an algorithmic goal includes minimizing the number of used links then we assume that each node knows its neighbors prior to the beginning of an execution, in having the neighbors correctly mapped on  ports.

The notation $\Lambda$ will denote an upper bound on the stretch of a communication network  in the course of an execution.
If $\Lambda$ is used then this means that $\Lambda$ is known to all nodes.

\Paragraph{Performance metrics.}

An initially connected graph $G$ evolves through a nested sequence of its original sub-graphs, by removing edges representing faulty links.
We want to assess how the running time an algorithm scales to the sub-networks resulting by removing faulty links.
Stretch is used as a parameter capturing the challenge of solving disconnected agreement in sub-networks.
The challenge in using stretch is that it evolves as a sequence of numbers.
Let the notation $\lambda_k$ mean a stretch of a network $G$ at a round~$k$ of an execution, and $\lambda$ be a stretch at the round of halting.

\begin{proposition}
\label{pro:stretch}

Stretches $\lambda_k$ make a monotonously increasing sequence: $\lambda_i\le \lambda_j$, for $i < j$.
\end{proposition}

\begin{proof}
Consider an edge $e$ belonging to a connected component~$C$.
If $e$ gets removed from the graph and the connected component $C$ stays intact then its diameter may only increase.
Suppose the edge~$e$ is a bridge of $C$, and after its removal the connected component~$C$ breaks into two connected components $C_1$ and~$C_2$.
The stretch of the graph after removal of the edge~$e$ is the stretch before removal of this edge minus the diameter of $C$ plus the diameters of $C_1$ and $C_2$, and this number incremented by~$1$.
So this new stretch cannot be smaller than the original one, because the diameter of $C$ is at most a sum of the diameters of $C_1$ and $C_2$ incremented by one.
\end{proof}

We consider bounds on performance of algorithms with respect to the stretch~$\lambda$ occurring at a  round in which all nodes halt.
Let $f:\mN\rightarrow \mN$ be a monotonously increasing function.
For a communication algorithm to run in $\cO(f(\lambda))$ time, it is sufficient and necessary that the running time is $\cO(f(\lambda_k))$, for any round number $k$ prior to halting, by Proposition~\ref{pro:stretch}.

A disconnected-agreement algorithm in a synchronous network with links prone to failures is \emph{early stopping}  if it runs in a number of rounds proportional to the unknown stretch~$\lambda$ actually occurring. Such an algorithm is \emph{fast}  if it runs in  a number of rounds proportional to an upper  bound on stretch~$\Lambda$ known to all the nodes.

Our use of terms ``fast'' and ``early stopping'' follows a similar approach to consensus algorithms with node crashes. 
Let $t$ denote an upper bound on the number of node crashes, which is known to all nodes, and $f$ be an actual number of node crashes, which is not assumed to be known.
Fast algorithms operate in running time $\cO(t+1)$ and early stopping algorithms operate in running time $\cO(f+1)$.
We use $\Lambda$ similarly as $t$ and $\lambda$ similarly as~$f$.

\section{Fast Agreement}

\label{sec:fast-agreement}

We demonstrate the relevance of stretch to the running time of algorithms solving disconnected agreement.
First, we show that a known upper bound on stretch allows to structure a disconnected agreement solution such that it operates in a number of rounds equal to the bound on stretch, so it is fast in our terminology.
Second, we show that stretch is a meaningful yardstick to measure such  running times, in that each  algorithm solving disconnected algorithm may have to be executed  for a number of rounds at least as large as the stretch.

If $\Lambda$ is an upper bound on the stretch known to all nodes, then the following could be a possible approach to reach disconnected agreement in networks prone to link failures.
Each node maintains a \emph{candidate value} in a dedicated variable, which it initializes to the input.
For $\Lambda$ rounds, each node does two things per round: (1) sends the current candidate value to all neighbors, (2) updates the candidate value to the maximum of the current value and the values just received in messages.
This algorithm is incorrect, which can be demonstrated on an example as follows.
Let a network consist of three nodes: node~$p_1$ with input~$1$, node~$p_2$ with input~$2$, and node~$p_3$ with input~$3$, all connected into a line $p_1-p_2-p_3$.
Suppose the link $p_1-p_2$ is non-faulty and $p_2-p_3$ is faulty.
The stretch is $2$, as the stretch of a network starting a line of $n$ nodes is always~$n-1$.
In the course of an execution, let the link $p_2-p_3$ not deliver messages in the first round and deliver messages in the second round. 
The nodes $p_1$ and $p_2$ are in the same connected component $p_1-p_2$, but $p_1$ decides on~$2$ while $p_2$ decides on~$3$, which violates the required agreement property.

We present next a fast algorithm solving disconnected agreement, assuming that a bound~$\Lambda$ on stretch is known to all nodes.
The algorithm is called \textsc{Fast-Agreement}; its pseudocode is given in Figure~\ref{fig:fast-agreement}.
Every node strives to eventually decide on the largest value possible, so a node remembers,  in a variable \texttt{candidate}, only the largest value it has witnessed so far.
This variable is initialized to the input value.
At each round, a node either sends its \texttt{candidate} value to all neighbors or pauses in sending messages, then receives all incoming messages from the neighbors, if any, and updates \texttt{candidate} to the maximum value received in messages is some are greater than the current value.
A node sends messages to neighbors at a round if either this is the first round, so \texttt{candidate} is the input value, or  if \texttt{candidate} was updated to a new value in the previous round.
We have that if a node sends a particular value to its neighbors then it does so only once.
All nodes halt after $\Lambda$ rounds, and every node decides on the \texttt{candidate} value, which is the largest value it has learned about.


\begin{figure}[t]

\hrule

\FF

\noindent
\texttt{algorithm} \textsc{Fast-Agreement} ($\Lambda$)

\FF

\hrule

\FF

\begin{enumerate}[nosep]
\item
initialize $\texttt{candidate} \leftarrow \texttt{input}_p$

\item
\texttt{repeat} $\Lambda$  \texttt{times}
\begin{itemize}[nosep]
\item[] 
\texttt{if} the current value of \texttt{candidate} has not been sent before \texttt{then} 
\begin{itemize}[nosep]
\item[]
send  \texttt{candidate} to all neighbors 
\end{itemize}
\item[] 
receive messages from all neighbors 

\item[] 
\texttt{if} a value greater than \texttt{candidate} has just been received 
\begin{itemize}[nosep]
\item[] 
\texttt{then} set \texttt{candidate} to the maximum value just received 
\end{itemize}
\end{itemize}
\item
decide on \texttt{candidate}
\end{enumerate}

\FF

\hrule

\FF

\caption{\label{fig:fast-agreement}
A pseudocode of algorithm \textsc{Fast-Agreement} for a node~$p$.
The parameter $\Lambda$ represents an upper bound on stretches, which is known to all nodes.}
\end{figure}

In analyzing correctness of the algorithm, we asses when values of some magnitude  carried in messages reach connected components of the network~$G_F$.
Let $C$ be a connected component of graph~$G_F$.
We say that \emph{value~$v$ reaches~$C$} at a round if a node in $C$ has its \texttt{candidate} value at least as large as~$v$ by this round.
Graph $G_F-C$ denotes $G_F$ with the nodes in $C$ removed along with their incident edges.

\begin{lemma}
\label{lem:chase}

If a value reaches a connected component $C$ of network $G_F$ then this occurs by the round equal to the stretch of $G_F-C$ plus~$1$.
\end{lemma}

\begin{proof}
The argument is by induction on the number of connected components of $G_F$.
The base case occurs if $G_F$ is connected.
Then $G_F=C$, so $G_F-C$ is an empty graph of stretch~$0$.
Any value that reaches $G_F$ is at least as large as the input of some node, and inputs are available by round~$1$.
Next comes the inductive step.
Let $C$ be a connected component of~$G_F$, and assume that the inductive hypothesis holds for~$G_F-C$.
If a value~$v$ has reached $C$ by the first round then we are done.
Otherwise, a value at least as large as~$v$ has been brought in a message sent by a node~$p_1$ to another node~$p_2$, where $p_1$ is in~$G_F-C$ and $p_2$ is in~$C$.
Let $C'$ be the connected component of $p_1$ in $G_F-C$. 
By the inductive assumption, $v$ reached $C'$ by a round $r$ equal to the stretch of $G_F-C-C'$ plus~$1$.
After round $r$, the value traveled through $C'$ for a number of rounds at most equal to the diameter of~$C'$, because a node transmits a value only in the round immediately following the round in which the value is received and becomes the new candidate value.
After the journey through $C'$, it takes one round to tranmit a value as large as~$v$ from~$p_1$ to~$p_2$, for a total number of rounds as large as the stretch of~$G_F-C$ plus~$1$.
This completes the inductive step, and so the argument by induction.
\end{proof}


\begin{theorem}
\label{thm:fast-agreement}

Consider an execution  of algorithm \textsc{Fast-Agreement}\,$(\Lambda)$ in a network.
If the stretch  of the network never gets greater than $\Lambda$ then the algorithm solves disconnected agreement in $\Lambda$ rounds using messages of $\cO(\log n)$ bits.
\end{theorem}

\begin{proof} 
The pseudocode in Figure~\ref{fig:fast-agreement} is structured as a loop of $\Lambda$ iterations, each iteration taking one round. 
All nodes decide at the end, which gives the termination and time performance.

Each node~$p$ initializes its private variable \texttt{candidate}$_p$  to the input value \texttt{input}$_p$.
An iteration of the repeat loop maintains an invariant that the variables \texttt{candidate} store only values taken from among the original input values, because only the values of this variable are sent and received.
Every message carries only a candidate value.
By the assumption on properties of input values, each such a value can be encoded with $\cO(\log n)$ bits.
Ultimately, each node~$p$ decides on a value in its \texttt{candidate}$_p$ variable, which gives validity.

Next, we show agreement.
Let us consider a connected component~$C$ of~$G_F$.
Let $v$ be the maximum value that has ever reached~$C$.
By Lemma~\ref{lem:chase}, this value has reached $C$ by the round equal to the stretch of~$G_F-C$ plus~$1$.
Once $v$ arrives at a node in~$C$, this connected component~$C$ is flooded with value~$v$, by design of the algorithm.
This takes time equal to the diameter of~$C$.
It follows that flooding of $C$ with $v$ is completed by the round equal to the stretch of~$G_F$.
This is the round all nodes decide, so all the nodes in $C$ decide on~$v$.
\end{proof}

\Paragraph{A lower bound on running time.}

We show a lower bound on the running time required to solve disconnected agreement in networks with link failures.
This lower bound equals the stretch of the final graph of the network.
We consider specific network topologies with the property that with no link failures the time needed to reach agreement is at least a diameter of the network.
If a final graph is connected, then the stretch is the same as the diameter.


\begin{figure}[t]
\begin{center}
\includegraphics[width=0.44\textwidth]{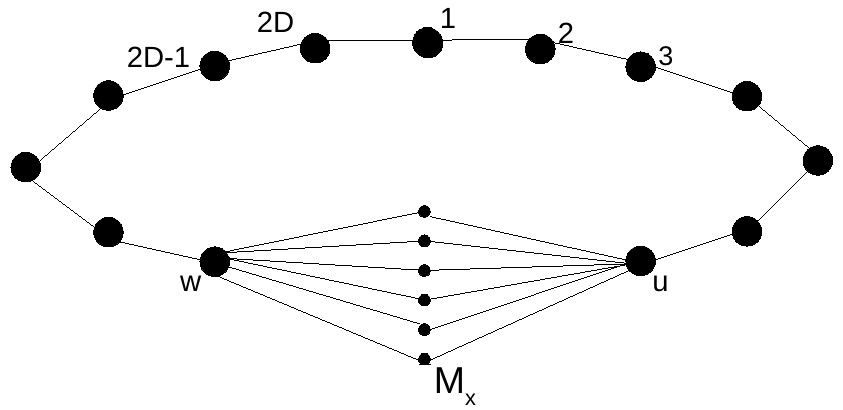}

\caption{\label{fig:kontr1} 
A depiction of graph $G(x,D)$ used in the proof of Lemma~\ref{lem:graph-diameter-no-failures}.}
\end{center}
\end{figure}


\begin{lemma}
\label{lem:graph-diameter-no-failures}

For any algorithm $\cA$ solving disconnected agreement in networks prone to link failures, and for positive integers $D$ and $n\ge 2D$, there exists a network $G$ with $n$ nodes and with diameter~$D$ such that some execution of $\cA$ on $G$ takes at least $D$ rounds with no link failures.
\end{lemma}

\begin{proof}
We define graphs $G(x,D)$, where $x>0$ is an integer.
Start with a cycle  of length $2D$, and  some three consecutive vertices $u$, $v$, and~$w$ on the cycle, where $v$ is connected by edges with $u$ and~$w$.
Replace vertex $v$ by $x$ copies of $v$, each connected precisely to the neighbors $u$ and $w$ of~$v$.   
This construction is depicted in Figure~\ref{fig:kontr1}, in which $M_{x}$ denotes $x$ copies of some vertex~$v$.
Observe that $G(x,D)$ has $2D+x-1$ vertices and diameter~$D$.
Take $G=G(x,D)$ such that  $n=2D+x-1$.

We consider executions of algorithm $\cA$ on a network modeled by~$G$ in which no link ever fails.
Each input value of a node in the network will be either $0$ or~$1$.
Start with the initial configuration~$C_0$ in which all input values are $0$, so decision is on~$0$ by validity.
We proceed through a sequence of initial configurations $C_0, C_1, C_2,\ldots, C_n$, such that $C_i$ has $i$ nodes starting with the input value~$1$ and $n-i$ nodes starting with the input value~$0$.
In particular, in configuration~$C_n$ all the nodes start with the input value~$1$, so decision has to be on~$1$,  by validity.
There exist two configurations $C_i$ and $C_{i+1}$ such that for~$C_i$ the decision is on~$0$ and for~$C_{i+1}$ the decision is on~$1$.
These two configurations differ only in some vertex $p$ having input~$0$ in $C_i$ and input $1$ in $C_{i+1}$.

Consider two executions of algorithm $\cA$, one starting in configuration~$C_i$ and the other in~$C_{i+1}$.
Each node~$q$  of distance~$D$ from~$p$ sends and receives the same messages in the first $D-1$ rounds in both executions. 
Take as $q$ a node that is of distance $D$ from $p$; such a node~$q$ exists by the specification of graph~$G$.
Node~$q$ needs to wait with deciding until at least the round~$D$, because its state transitions in the first $D-1$ rounds are the same in both executions, while the decisions are different.
\end{proof}

\newpage


\begin{corollary}
\label{cor:n-half}

For any algorithm $\cA$ solving disconnected agreement in networks prone to link failures, and for any even positive integer~$n$, there exists a network $G$ with $n$ nodes  such that some execution of $\cA$ on $G$ takes at least $\frac{n}{2}$ rounds with no link failures.
\end{corollary}

\begin{proof}
We use a network modeled by graph $G(x,D)$ used in the proof of Lemma~\ref{lem:graph-diameter-no-failures}, in which $x=1$ and $n=2D$.
We have that $D=\frac{n}{2}$ rounds are necessary to reach agreement in some executions of algorithm~$\cA$, by Lemma~\ref{lem:graph-diameter-no-failures}. 
\end{proof}


\begin{theorem}
\label{thm:lower-bound-stretch}

For any natural number $\lambda$ and an algorithm $\cA$ solving disconnected agreement in networks prone to link failures, there exists a network $G$ that has stretch $\lambda$ and such that each execution of $\cA$ on $G$ takes at least $\lambda$ rounds.
\end{theorem}

\begin{proof} 
Let us take a network $G$ with $n\ge 2\lambda$ nodes and with diameter~$\lambda$ such that some execution of algorithm~$\cA$ on~$G$ takes at least $\lambda$ rounds with no link failures.
Such a connected network exists by Lemma~\ref{lem:graph-diameter-no-failures}.
The stretch of a connected network equals its diameter.
\end{proof}

Theorem~\ref{thm:lower-bound-stretch} shows that algorithm \textsc{Fast-Agreement}($\Lambda)$ is asymptotically time-optimal on networks prone to link failures, because the actual stretch $\lambda$ could be as large as an upper bound~$\Lambda$ on stretch.

\section{General Agreement with Short Messages}

\label{sec:minimal-knowledge}

We present a general disconnected-agreement algorithm using short  messages of $\cO(\log n)$ bits.
Algorithm \textsc{Fast-Agreement} presented in Section~\ref{sec:fast-agreement}, which also employs messages of $\cO(\log n)$ bits, relies on an upper bound on stretch $\Lambda$ that is a part of code, and if the actual stretch in an execution goes beyond $\Lambda$ then an execution of algorithm  \textsc{Fast-Agreement} may not be correct.
We assume in this section that nodes rely on minimal knowledge only and the given algorithm is correct for arbitrary patterns of link failures and the resulting stretches.
The algorithm terminates in at most $n$ rounds, while the number of nodes~$n$ is not known.
The running time is asymptotically optimal in case there are no failures, by Corollary~\ref{cor:n-half} in Section~\ref{sec:fast-agreement}.

The following approach could possibly result in reaching disconnected agreement in networks prone to link failures, while sending only short messages.
A node~$p$ maintains a set of input values in a dedicated variable, which it initializes to the singleton set with the node's input.
The node~$p$ does two things per round: (1) for each neighbor, $p$ sends some input value to this neighbor, but only such it has never sent before, (2) updates the set of known inputs value by adding new values just received in messages.
A node halts after executing as many rounds as the size of the set of input values and decides on the maximum input value known.
This algorithm works correctly if all the $n$ input values are distinct, because each node collects all the input values in its connected component in the final graph.
This algorithm is incorrect in general, which can be demonstrated on an example of a network of three nodes: node~$p_1$ with input~$1$, node~$p_2$ with input~$1$, and node~$p_3$ with input~$3$, connected into a line $p_1-p_2-p_3$.
To see incorrectness, suppose both links $p_1-p_2$ and $p_2-p_3$ are non-faulty.
The node $p_1$ halts at the second round and decides on~$1$ and node $p_3$ halts at the third round and decides on~$3$, while all the nodes are in the same connected component.

Next, we give an algorithm for disconnected agreement that uses short messages while the nodes do not know the size~$n$ of the network nor any bound on stretch.
Each node~$p$ has its name stored as $\texttt{name}_p$ and initial input value stored at $\texttt{input}_p$.
A node~$p$ forms a pair $(\texttt{name}_p,\texttt{input}_p)$, which we call \emph{input pair}, and sends it through each port as the first communication in an execution; simultaneously node~$p$ receives input pairs from its neighbors.
Every other node $q$ forms its respective pair $(\texttt{name}_q,\texttt{input}_q)$.
If $p\ne q$ then also $(\texttt{name}_p,\texttt{input}_p)\ne (\texttt{name}_q,\texttt{input}_q)$, even if $\texttt{input}_p=\texttt{input}_q$, so there are $n$ input pairs in total.
As an execution continues, nodes exchange input pairs  among themselves.
A message carries one input pair.
The goal for each node is to learn as many input pairs as possible and help other nodes in accomplishing the same goal.
A node halts in the first round whose number surpasses the number of input pairs the node has learned.

Next, we discuss how the algorithm is implemented.
The algorithm is called \textsc{SM-Agreement}, its pseudocode is in Figure~\ref{fig:alg-short-messages}.
A node maintains a list \texttt{Inputs} of all input pairs that the node has ever learned about. 
For each port~$\alpha$, a node maintains a set $\texttt{Channel}[\alpha]$ storing all input pairs that have either  been received through $\alpha$ or sent through $\alpha$.
At every round, a node verifies for each port~$\alpha$ if there is a message to be sent through the port.
This is determined by examining the list \texttt{Inputs} to check if there is an input pair in the list  that does not belong to the set $\texttt{Channel}[\alpha]$: if this is the case then the first such a pair is transmitted through $\alpha$.
At every round, a node may receive a message through a port~$\alpha$: if this occurs, the node adds the received input pair to set $\texttt{Channel}[\alpha]$, unless the pair is already there, and appends the pair to the list \texttt{Inputs}, unless the pair is already in the list.
A node maintains a counter of rounds called \texttt{round}, which is incremented in each round.
An execution ends when this counter  surpasses the size of the list \texttt{Inputs}.


\begin{figure}[t]

\hrule

\FF

\texttt{algorithm} \textsc{SM-Agreement} 

\FF

\hrule

\FF

\begin{enumerate}[nosep]
\item 
initialize: \texttt{Inputs} to empty list, $\texttt{round} \leftarrow 1$ ; append $(\texttt{name}_p, \texttt{input}_p)$ to \texttt{Inputs}  
\item
\texttt{for} each port $\alpha$ \texttt{do}
\begin{enumerate}[nosep]
\item[] 
initialize set $\texttt{Channel}[\alpha]$ to empty ; send $(\texttt{name}_p, \texttt{input}_p)$ through $\alpha$
\end{enumerate}
\item
\texttt{for} each port $\alpha$ \texttt{do}
\begin{enumerate}[nosep]
\item[] 
if a pair $(\texttt{name}_q, \texttt{input}_q)$ received through $\alpha$ \texttt{then} 
\begin{enumerate}[nosep]
\item[] 
add $(\texttt{name}_q, \texttt{input}_q)$ to $\texttt{Channel}[\alpha]$ ; append $(\texttt{name}_q, \texttt{input}_q)$ to \texttt{Inputs}
\end{enumerate}
\end{enumerate}
\item \label{SM-repeat-loop}
\texttt{repeat}
\begin{enumerate}[nosep]
\item 
\label{instr:SM-send}
\texttt{for} each port $\alpha$ \texttt{do}
\begin{enumerate}[nosep]
\item[]
\texttt{if} some item in \texttt{Inputs} is not in \texttt{Channel}$[\alpha]$ \texttt{then}
\begin{itemize}[nosep]
\item[]
 let $x$ be the first such an item ; send $x$ through $\alpha$ ; add $x$ to \texttt{Channel}$[\alpha]$
\end{itemize}
\end{enumerate}
\item 
\label{instr:SM-receive}
\texttt{for} each port $\alpha$ \texttt{do}
\begin{enumerate}[nosep]
\item[]
if a pair $(\texttt{name}_q, \texttt{input}_q)$ was just received through $\alpha$ \texttt{then} 
\begin{enumerate}[nosep]
\item[] 
add $(\texttt{name}_q, \texttt{input}_q)$ to $\texttt{Channel}[\alpha]$ ; append $(\texttt{name}_q, \texttt{input}_q)$ to \texttt{Inputs}
\end{enumerate}
\end{enumerate}
\item 
$\texttt{round} \leftarrow \texttt{round} +1$

\end{enumerate}
\texttt{until} $\texttt{round} > |\texttt{Inputs}|$ 
\item 
decide on the maximum input value in \texttt{Inputs}
\end{enumerate}

\FF

\hrule

\FF

\caption{\label{fig:alg-short-messages}
A pseudocode for a node~$p$. 
The operation of adding an item to a set is void if the item is already in the set.
The operation of appending an item to a list is void if the item is already in the list.
The notation $|\texttt{Inputs}|$ means the number of items in the list \texttt{Inputs}.}
\end{figure}

We may consider a modification of algorithm \textsc{SM-Agreement}, as presented in Figure~\ref{fig:alg-short-messages}, in which the repeat loop \eqref{SM-repeat-loop} is iterated indefinitely, rather than stopping after $|\texttt{Inputs}|$ iterations.
In an execution of the modified algorithm, eventually lists \texttt{Inputs} and sets \texttt{Channel}$[\alpha]$, for all ports~$\alpha$, stabilize in all nodes, so that messages are no longer sent.
Suppose two nodes $p$ and $q$ are connected by a reliable path at a round when node~$q$ decides.
By the round all the variables stabilize, the input pair $(\texttt{name}_r,\texttt{input}_r)$ will have reached node~$q$, because there is a reliable path from $p$ to~$q$.
We need to show that this also occurs in an execution of algorithm \textsc{SM-Agreement} without the modification of the condition controlling the repeat loop, for the same link failures in the two executions.
To this end, it suffices to demonstrate that node~$q$ withholds deciding long enough.

\begin{lemma}
\label{lem:SM-agreement}

If nodes $p$ and $q$ are connected by a reliable path at a round when node~$q$ decides, and the  list $\texttt{Inputs}_{p}$ includes an input pair $(\texttt{name}_r,\texttt{input}_r)$ at that round, then the  list $\texttt{Inputs}_{q}$  includes this pair $(\texttt{name}_r,\texttt{input}_r)$ at the round when $q$ decides.
\end{lemma}

\begin{proof}
Let $\gamma$ be a path between nodes $r$ and $p$ through which pair $(\texttt{name}_r,\texttt{input}_r)$ arrives to node~$p$ first; suppose $\gamma = (s_1,s_2,\ldots, s_k)$, where $r=s_1$ and $s_k=p$.
Let $\delta$ be a reliable path connecting node~$p$ to $q$ at a round when node~$q$ decides; suppose $\delta = (t_1,t_2,\ldots, t_\ell)$, where $p=t_1$ and $t_\ell=q$.
We want to show that all input pairs originating at the nodes on the paths~$\gamma$ and~$\delta$ reach node~$q$ by the round in which $q$  decides.
Consider a path $\zeta$ obtained by concatenating path~$\gamma$ with path~$\delta$ at the node~$p=s_k=t_1$.
Let us denote the nodes on this path as $\zeta=(u_1,\ldots,u_\ell)$, where $u_1=r$ and $u_\ell=q$, and denote by~$v_i$ a pair of the form $(\texttt{name},\texttt{input})$ originating at~$u_i$, for reference.

The simplest scenario, for node~$q$ to receive input pair $v_1=(\texttt{name}_r,\texttt{input}_r)$, occurs when pair~$v_1$ gets sent in each consecutive round along the path~$\zeta$ until it reaches node~$q$.
During these transmissions, node~$q$ receives consecutive pairs~$v_i$, starting  from $v_\ell=(\texttt{name}_q,\texttt{input}_q)$ originating in~$q$, through $v_1=(\texttt{name}_r,\texttt{input}_r)$. 
In this scenario, a new pair gets added to list \texttt{Inputs}$_q$ in each of these rounds, which makes node~$q$ postpone deciding for a total of at least $\ell$ rounds.
We call this the \emph{basic} scenario.
Let us consider conceptual queues $Q_i$ at each of the nodes $u_i$ on $\zeta$, for $i<\ell$, where the queue~$Q_i$ stores pairs that still need to be sent to~$u_{i+1}$.
These queues are manipulated by actions specified in instructions~\eqref{instr:SM-send} and~\eqref{instr:SM-receive} in the pseudocode in Figure~\ref{fig:alg-short-messages}.
In the basic scenario, each queue~$Q_i$ is initialized with~$u_i$, and then each pair~$v_i$ moves through the queues~$Q_j$, for $j=i+1, i+2,\ldots$ during consecutive rounds, to eventually arrive at~$u_\ell=q$.

Next, we consider two possible alternative ways for input pairs to be sent along $\zeta$ to arrive at~$q$, while departing from the basic scenario.

One possible departure from the basic scenario occurs when a pair~$v_i$ gets delivered to~$u_j$, for $j>i$, by a shortcut outside of $\zeta$ rather than through consecutive nodes $u_{i+1}, u_{i+2},\ldots, u_j$.
This may result in pair $v_i$ reaching $q$ earlier than in the basic scenario, while postponing delivery of other pairs by one round.
What does not change is that a new pair gets delivered to $q$ at each round, so $q$ keeps waiting, and also $v_i$ gets added to~$Q_k$ to be later removed from~$Q_k$ exactly once, at each node~$u_k$ on the path~$\zeta$, for $k\ge i$, while sent down the path towards~$q$, by the rules of manipulating lists \texttt{Inputs} and sets \texttt{Channel} specified in instructions~\eqref{instr:SM-send} and~\eqref{instr:SM-receive} in the pseudocode in Figure~\ref{fig:alg-short-messages}.
In this scenario, pair $(\texttt{name}_r,\texttt{input}_r)$ gets delivered to $q$ no later than in the basic scenario.

Another possible departure occurs when a pair~$v$ originating at a node outside $\zeta$ gets delivered to a node on $\zeta$ and then travels along $\zeta$ towards~$q$.
Each such a pair $v$ simply increases the number of pairs transmitted along $\zeta$, so it may delay $v_1=(\texttt{name}_r,\texttt{input}_r)$ from reaching $q$ by one round.
At the same time, such a pair $v$ increments the size of the list \texttt{Inputs} at $q$ by one when added to it,  thereby extending the time period until $q$ decides by one round.
This makes node~$q$ wait long enough to receive input pair $(\texttt{name}_r,\texttt{input}_r)$ by the round it decides.

The two possible departures from the basic scenario discussed above can both occur in an execution, and also there could be multiple instances of pairs creating such scenarios departing from the basic one, without affecting the conclusion that node~$q$ waits long enough to receive pair $(\texttt{name}_r,\texttt{input}_r)$ prior to deciding.
\end{proof}


\begin{theorem}
\label{thm:SM-Agreement}

Algorithm {\sc SM-Agreement} solves disconnected agreement in $n+1$ rounds relying on minimal knowledge and using short messages each of $\cO(\log n)$ bits.
\end{theorem}

\begin{proof} 
Input pairs stored in the list \texttt{Inputs} at a node have names of nodes as their first components, so there can be at most $n$ pairs ever added to this list.
A node decides immediately when a round number exceeds the size of list~\texttt{Inputs}, which occurs by round~$n+1$.
This gives termination and the bound on running time.
A node decides on an input value in an input pair present  in its list \texttt{Inputs}, which stores pairs that originated at nodes of the network; this implies validity.
By Lemma~\ref{lem:SM-agreement}, all nodes in a connected component of the final graph~$G_F$ have identical lists \texttt{Inputs} at a round of deciding. 
A decision is on a maximum value in a list, which is uniquely determined by the set of pairs stored in the list; this gives  agreement. 
A message carries one input pair, so it consists of $\cO(\log n)$ bits, per the assumptions about the node names and input values.
\end{proof}

\section{Agreement with Linear Messages}

\label{sec:linear-messages}

The goal of this section is to develop an algorithm whose running time scales well to the stretch actually occurring in an execution.
We are ready to use messages longer than short ones used in the previous sections, and will use linear messages of $\cO(n\log n)$ bits.
Nodes are to rely on the minimal knowledge only: each node knows its  own name and can distinguish  ports by their communication functionality. 
We give an algorithm that works in $(\lambda+2)^3$ rounds, where $\lambda$ is the stretch at the round of halting. 
The size of linear messages imposes constrains on the design of algorithms, and the obtained algorithm is not early stopping, but its running time is polynomial in~$\lambda$.

A message of $\cO(n\log n)$ bits allows to represent any set of node names.
Once a node receives a message with names of nodes it can be certain that information about the magnitude of input values from these nodes has been received as well, assuming some relevant information about the input values was sent along.
In principle, a node may halt as soon as it has heard from all the nodes in its connected component.
The challenge is to detect such a moment, because a message of $\cO(n\log n)$ bits can bring in many node names at once at a round, but still not all of them, and that round may be followed by a long stretch of rounds without receiving new names.
To determine when to halt, we use timestamps, which are set to the current round number when created.
A message of $\cO(n\log n)$ bits can represent a set of pairs consisting of a node's name and a timestamp, as long as there is at most one such a pair for each node and the timestamp represents a round during an execution, while the running time of the algorithm is polynomial in~$n$.

\Paragraph{An overview of the algorithm.}

Every node maintains a counter of round numbers, incremented when a round begins.
In each round, a node~$p$ generates a new timestamp~$r$ equal to the current value of the round counter, and forms a pair $(\texttt{name}_p, r)$, which we call a \emph{timestamp pair} of node~$p$.
Such timestamp pairs are sent to the neighbors, to be forwarded through the network.
Each node node~$p$ stores a timestamp pair with the latest timestamp for a node it has ever received a timestamp pair from, and sends all such pairs to the neighbors in every round.  

An execution of the algorithm at a node is partitioned into \emph{epochs}, each epoch being a contiguous interval of rounds.
Epochs are not coordinated among nodes, and each node governs its own epochs.
The first epoch begins at round zero, and for the following epochs, the last round of an epoch is remembered in order to discern timestamp pairs sent in the following epochs.
For the purpose of monitoring progress of discovering the nodes in the connected component during an epoch, each node maintains a separate collection of timestamp pairs, which we call \emph{pairs serving the epoch}.
This collection stores only timestamp pairs sent in the current epoch, a pair with the greatest  timestamp per node which originally generated the pair.

The \emph{status of a node $q$ at a node~$p$} during an epoch can be either absent, updated, or stale.
If the node~$p$ does not have a timestamp pair for~$q$ serving the epoch then $q$ is \emph{absent} at~$p$.
If at a round of an epoch the node~$p$ either adds a timestamp pair serving the epoch for an absent node~$q$ or replaces a timestamp pair of a node~$q$ by a new timestamp pair with a greater timestamp than the previously held one, then $q$ is \emph{updated} at this round.
If the node~$p$ has a timestamp pair for a node~$q$ serving the epoch but does not replace it at a round with a different timestamp pair to make it updated, then $q$ is \emph{stale} at this round.
If the node~$p$ at a round $t_1$ receives  an epoch-serving timestamp pair $(\texttt{name}_q, t_2)$ for a node~$q$ that replaces a previously stored timestamp for~$q$ then the number $t_1-t_2$ is called the \emph{range of $q$ at~$p$}. 
Ranges are determined only for updated nodes at~$p$, since $q$ becomes updated after $p$ receives $(\texttt{name}_q, t_2)$.  
A range of~$q$ is the length of a path traversed by a timestamp pair of~$q$ to reach~$p$ in the epoch, at the round the timestamp pair is received, so it is a lower bound on the distance from~$q$ to~$p$ at this round.

We say that an epoch of a node~$p$ \emph{stabilizes} at a round if either no new node has its status changed from absent to updated at~$p$ or no node gets its range changed at~$p$.
If an epoch stabilizes at a round, then the epoch ends.
During an epoch, a node $p$ builds a set of names of nodes from which it has received timestamp pairs serving this epoch. 
A similar set produced in the previous epoch is also stored.
As an epoch ends, $p$ compares the two sets.
If they are equal then $p$ stops executing epochs, decides on the maximum input value  ever learned about, notifies the neighbors of the decision, and halts.


\begin{figure}[t]
\hrule
\FF
\texttt{algorithm} \textsc{LM-Agreement} 
\FF
\hrule
\FF
\begin{enumerate}[nosep]
\item \label{LM:initialization}	
initialize: 
$\texttt{candidate}_p \leftarrow \texttt{input}_p$ ,  
$\texttt{round} \leftarrow 0$,
$\texttt{Timestamps} \leftarrow \emptyset$,
$\texttt{Nodes} \leftarrow \perp$
\item  \label{LM:outer-repeat-loop}
\texttt{repeat}
\begin{enumerate}[nosep]
\item \label{LM:PreviousNodes-initialization}
$\texttt{epoch} \leftarrow \texttt{round}$, 
$ \texttt{PreviousNodes}  \leftarrow \texttt{Nodes} $,
$\texttt{EpochTimestamps} \leftarrow \emptyset$
\item 
\texttt{repeat}  \label{LM:inner-repeat-loop}
\begin{enumerate}[nosep]
\item \label{LM:round-increment}
$\texttt{round} \leftarrow \texttt{round} + 1$
\item 
add pair $(\texttt{name}_p, \texttt{round})$ to sets \texttt{Timestamps} and \texttt{EpochTimestamps}
\item \label{LM:for-each-port-do}
\texttt{for} each port  \texttt{do}
\begin{enumerate}[nosep]
\item
send \texttt{Timestamps} and $(\textrm{this-is-candidate},\texttt{candidate}_p)$ through the port 
\item
receive messages coming through the port 
\end{enumerate}
\item \label{LM:candidate}
\texttt{for} each received pair $(\textrm{this-is-candidate},x)$ \texttt{do}
\begin{enumerate}[nosep]
\item[]
\texttt{if}  $x > \texttt{candidate}_p$ \texttt{then} assign $\texttt{candidate}_p\leftarrow x$ 
\end{enumerate}
\item 
\texttt{for} each received timestamp pair  $(\texttt{name}_q,y)$ \texttt{do}
\begin{enumerate}[nosep]
\item
add $(\texttt{name}_q,y)$ to \texttt{Timestamps} if this is a good update 
\item
\texttt{if} $y> \texttt{epoch}$ \texttt{then} add $(\texttt{name}_q,y)$ to \texttt{EpochTimestamps} if this is a good update
\end{enumerate}
\end{enumerate}	
\item \label{LM:control-inner-loop}
\texttt{until} epoch stabilized at the round
\item \label{LM:determine-Nodes}
set \texttt{Nodes} to the set of first coordinates of timestamp pairs in \texttt{EpochTimestamps}
\end{enumerate}
\item \label{LM:inner-until-condition}
\texttt{until} $ \texttt{PreviousNodes} = \texttt{Nodes} $ 
\item  
send $(\textrm{this-is-decision},\texttt{candidate}_p)$ through each port
\item \label{LM:deciding} 
decide on $\texttt{candidate}_p$ 
\end{enumerate}
\FF
\hrule
\caption{\label{fig:lm-agreement}
A pseudocode for a node~$p$.
Each iteration of the main repeat-loop~\eqref{LM:outer-repeat-loop} makes an epoch.
Symbol $\perp$ denotes a value different from any actual set of nodes, so the initialization of \texttt{Nodes} to~$\perp$ in line~\eqref{LM:initialization} guarantees execution of at least two epochs.
A \emph{good update} of a timestamp pair for a node $q$ either adds a first such a pair for $q$ or replaces a present pair for~$q$ with one with a greater timestamp.
At each round, $p$ checks to see if a message of the form $(\textrm{this-is-decision},z)$ has been  received, and if so then $p$ forwards this message through each port, then decides on~$z$, and halts. 
}
\end{figure}

\Paragraph{An implementation of the algorithm.}

The algorithm is called \textsc{LM-Agreement}, its pseudocode is given  in Figure~\ref{fig:lm-agreement}.
An execution starts with initialization of some variables by instruction~\eqref{LM:initialization}.
The main repeat loop follows as instruction~\eqref{LM:outer-repeat-loop};  one iteration of the main repeat-loop makes an epoch.
The pseudocode refers to a number of variables; we review them next.

Each node~$p$ uses a variable \texttt{candidate}$_p$, which it initializes to $\texttt{input}_p$.
Node~$p$ creates a pair $(\textrm{this-is-candidate},\texttt{candidate}_p)$, which we call a \emph{candidate pair of~$p$}.
Nodes keep forwarding their candidate pairs to the neighbors continually.
If a node~$p$ receives a candidate pair of some other node with a value~$x$ such that $x>\texttt{candidate}_p$ then $p$ sets its $\texttt{candidate}_p$ to~$x$.
An execution concludes with deciding by performing instruction~\eqref{LM:deciding}.
Just before deciding, a node notifies the neighbors of the decision.
Once a notification of a decision is received, the recipient forwards the decision to its neighbors, decides on the same value, and halts.

The variable \texttt{round} is an integer counter of rounds, which is incremented in each iteration of the inner repeat loop by executing instruction~\eqref{LM:round-increment}.
The round counter is used to generate timestamps.
The variable \texttt{Timestamps} stores timestamp pairs that $p$ has received and forwards to its neighbors.
The variable \texttt{EpochTimestamps} stores timestamp pairs serving the current epoch, which have been  generated after the beginning of the current epoch.
Each set \texttt{Timestamps} and \texttt{EpochTimestamps} stores at most one timestamp pair per node, the one with the greatest received timestamp.
Each iteration of the inner repeat loop~\eqref{LM:inner-repeat-loop} implements one round of sending and collecting messages through all the ports by executing instruction~\eqref{LM:for-each-port-do}.
The inner repeat loop~\eqref{LM:inner-repeat-loop} ends as soon as the epoch stays stable at a round, which is represented by condition~\eqref{LM:control-inner-loop}.

The variable \texttt{Nodes} stores the names of nodes from which timestamp pairs serving  the epoch have been received.
The variable \texttt{Nodes} is calculated at the end of an epoch by instruction~\eqref{LM:determine-Nodes}.
The set of nodes in \texttt{Nodes} at the end of an epoch is stored as \texttt{PreviousNodes} at the start of the next epoch.
The main repeat loop~\eqref{LM:outer-repeat-loop} stops to be iterated as soon as the set of names of nodes stored in \texttt{Nodes} stays the same as the set stored in \texttt{PreviousNodes}, which is checked by condition~\eqref{LM:inner-until-condition}.
We want the set \texttt{Nodes} to be calculated  in two consecutive epochs at least once.
To guarantee this, \texttt{PreviousNodes} is initialized to a special value denoted~$\perp$ by instruction~\eqref{LM:PreviousNodes-initialization} in the pseudocode, when the instruction is executed for the first time.
The value~$\perp$ is defined by the property that it is different from any set of nodes.

\Paragraph{The correctness and running time.}

A node~$p$ is said to have \emph{heard of node~$q$} in an epoch, if the node~$p$ received a timestamp pair from~$q$ serving this epoch.
We mean an epoch according to the node~$p$, since epochs are not coordinated across the network, and so a different node~$q$ could be in a different epoch by its count of iterations of the main repeat loop~\eqref{LM:outer-repeat-loop} in the pseudocode in Figure~\ref{fig:lm-agreement}.

\begin{lemma}
\label{lem:LM-radius}

If a node~$p$ has not heard of some node  by a round of an epoch yet and that node is still connected to~$p$ by a reliable path at the next round, then the epoch continues through the next round.
\end{lemma}

\begin{proof}
The proof is by induction on round numbers in an epoch.
The node~$p$ starts an epoch with an empty set of timestamps pairs serving the epoch and adds at least its own timestamp pair to make itself updated at the first round of the epoch.
If the node~$p$ has neighbors connected to it by reliable links, then, at the first round  of the  epoch, $p$ hears of them and their status becomes updated as well.
This means that an epoch never stabilizes at the first round, so it always continues beyond the first round.
This provides the base step of induction.

Next we consider the inductive step.
Let a round~$i$ of an epoch executed by the node~$p$ be such that node~$p$ has not heard of some node~$q$ in its current connected component up to round~$i-1$, and that such a node~$q$ stays connected to~$p$ during round~$i$.
If the node~$p$ has not heard of the  node~$q$ by round~$i-1$, then $q$'s distance from~$p$ is greater than~$i-1$ at the beginning of round~$i$, so the distance is at least~$i$.
Let $\pi=(s_0, s_1,\ldots, s_\ell)$ be a shortest path from~$p$ to~$q$, that exists at round~$i$, where $p=s_0$ and $ s_\ell=q$, for $\ell\ge i$.
Node~$s_{i}$ has its timestamp pair delivered by $s_1$ to $p=s_0$ at round~$i$, because such a timestamp pair has just completed its traversal of the path~$\pi$ towards~$p$.
This arrival establishes the range of~$s_{i}$ as~$i$, because this is the distance to~$p$ at this round.
If the node~$p$ has not heard of~$s_{i}$ before, then $p$ changes the status of $s_{i}$ from absent to updated.
If the node~$p$ has heard of~$s_{i}$ before, then the range of $s_{i}$ at $p$ is at most~$i-1$ after round~$i-1$, so it gets changed to~$i$.
The epoch does not stabilize during round~$i$ in either case, so it continues through the next round $i+1$.
\end{proof}

Each epoch at a node~$p$ eventually comes to an end.
This is because eventually all nodes in the connected component of $p$ are either updated or stale at~$p$, and their ranges stabilize to ones resulting from link failures that manifested themselves in the execution up to this point.

\begin{lemma}
\label{lem:LM-epoch}

The duration of an epoch at every node  is at most $(\lambda+2)^2$.
\end{lemma}

\begin{proof}
Let us consider a node $p$.
At a round $i$ of an epoch, the status of every node of distance at most~$i$ from~$p$ at round~$i$ becomes either updated or stale at node~$p$.
Eventually a round $k$ occurs such that the diameter of the connected component of $p$ is $k$, and after this round node~$p$ will have heard of every node in its connected component.
If a node~$q$ is of distance~$i$ from $p$ at a round~$j$, and this distance stays equal to~$i$ for at least $i$ rounds following round~$j$, then the range of every node on a shortest path of length~$i$ from~$p$ to~$q$ gets updated to its current value, which happens by round~$j+i$. 
Let $q$ be an arbitrary node.
The distance from~$q$ to node~$p$ can change at most as many times as the diameter of the connected component of~$p$ at the round $p$ ends the epoch.
After each such a change, it takes up to as many rounds as the distance from~$p$ to~$q$ to have ranges of nodes on a shortest path from~$q$ to~$p$ updated to new values.

Let $D$ denote the diameter of the connected component of node~$p$ when it ends the epoch.
If $D=0$ then the epoch ends after two rounds.
If $D\ge 1$ then the round in which $p$ ends the epoch is at most
\[
D + \sum_{i=1}^{D-1} i^2 = D+ \frac{D(D-1)}{2}\le D^2
\ .
\]
It follows that an epoch takes at most $(\lambda+2)^2$ rounds
\end{proof}

 
\begin{theorem}
\label{thm:lm-agreement}

Algorithm \textsc{LM-Agreement} solves disconnected agreement in $(\lambda+2)^3$ rounds relying on minimal knowledge and using messages of $\cO(n\log{n})$ bits.
\end{theorem}

\begin{proof} 
If a node~$p$ decides then the decision is on its candidate value stored in the variable \texttt{candidate}.
Such a decision value is the initial input value of some node, by the instructions~\eqref{LM:initialization} and~\eqref{LM:candidate}  in Figure~\ref{fig:lm-agreement}.
This gives validity.

For agreement, we argue that when a node~$p$ decides, then it knows the maximum candidate value of all the nodes in its connected component at the round of deciding.
The equality of sets \texttt{Nodes} and \texttt{PreviousNodes}, verified as condition~\eqref{LM:inner-until-condition} in the pseudocode in Figure~\ref{fig:lm-agreement} at the end of the epoch in which $p$ decides, guarantees that the nodes of the connected components at the ends of the current and previous epochs have stayed the same.
Consider the maximum candidate value present among the nodes of the connected component of $p$ at the round of its deciding. 
This candidate value was also maximum among the values held by nodes of the connected component of~$p$ in the previous epoch, since no nodes got disconnected from~$p$ in the current epoch.
Suppose the maximum candidate value was held by a node~$q$ in the connected component at the end of the previous epoch. 
The node~$p$ hears from~$q$ in the current epoch, by Lemma~\ref{lem:LM-radius}.
The candidate value of~$q$ travels along the same paths to~$p$ as the timestamp pairs from~$q$.
It follows that the node~$p$ decides on the candidate value of the node~$q$, and this candidate value is maximum among candidate values of all nodes in the connected component at the round of deciding.

Next, we estimate the running time.
A connected component of a node~$p$ can evolve through a sequence of contractions, occurring when the connected component shrinks and some nodes get disconnected to make other connected components. 
The number of epochs for every node is at least two, and it is at most the number of connected components of the final graph plus~$1$, for which the stretch plus~$1$ is an upper bound.
An epoch at a node~$p$ takes at most $(\lambda+2)^2$ time, by~Lemma~\ref{lem:LM-epoch}.
The number of rounds by halting for a node is a sum of lengths of its epochs, so it is at most~$(\lambda+2)^3$.

Finally, we estimate the size of messages.
A node sends its variable \texttt{Timestamps} and a candidate pair $(\textrm{this-is-candidate},\texttt{candidate})$ in a message.
The set \texttt{Timestamps} includes at most one timestamp pair per node.
A node's name needs $\cO(\log n)$ bits and a timestamp needs at most $\lg n^3=\cO(\log n)$ bits, because $\lambda <n$.
Each candidate value is some original input of a node, so it also needs $\cO(\log n)$ bits.
\end{proof}

\section{Early Stopping Agreement}

\label{sec:early-stopping}

We give an early-stopping disconnected agreement algorithm whose running time performance~$\cO(\lambda)$ scales optimally to the stretch $\lambda$ occurring in an execution by the time of halting. 
Nodes rely only on the minimal knowledge, similarly as in algorithms \textsc{SM-Agreement} (in Section~\ref{sec:minimal-knowledge}) and \textsc{LM-Agreement} (in Section~\ref{sec:linear-messages}), but messages  carry $\cO(m\log n)$ bits.
This size is greater than that of short messages with $\cO(\log n)$ bits in algorithm~\textsc{SM-Agreement} and linear messages with $\cO(n\log n)$ bits in algorithm~\textsc{LM-Agreement}.

\Paragraph{An overview of the algorithm.}

A graph whose vertices represent nodes and edges stand for links is like a map of the network at a round.
Nodes executing the algorithm keep sending their knowledge of the network's map to neighbors, while simultaneously receiving similar information from them.
The goal for each node is to build an approximation of a map of the network, which we call a snapshot.
More precisely, a \emph{snapshot} of the network at a node~$p$ consists of the names of nodes as vertices and edges between vertices representing links, but only these nodes and links that $p$ has heard about.
A snapshot at a node~$p$ may include the initial input for a vertex, should the node~$p$ know the input of a node represented by this vertex.
A whole snapshot encoded as a message takes $\cO(m\log n)$ bits.
We explain how nodes manage snapshots next.

The model of minimal knowledge means that initially a node knows only its own name and input, but its ports are not  labeled with names of the respective neighbors.
To mitigate this,  at the first round, every node sends its name through all the communication ports.
At this very round, nodes receive the names of their respective neighbors coming through communication ports.
This results in every node discovering its neighbors, which they use to assign neighbors' names to ports.
During the first round, nodes do not send their input values. 
After the first round, a node has a snapshot consisting of its own name, the names of neighbors that submitted their names, and edges connecting the node to its newly discovered neighbors.

Starting from the second round, nodes iterate the following per round: send the current snapshot to each neighbor, receive snapshots from the neighbors, and update the snapshot by incorporating newly acquired knowledge.		 
Links may be marked as unreliable, if they have failed at least once to deliver a transmitted message.
Edges representing unreliable links get removed from the snapshot. 
Such a removal is permanent, in that an edge representing an unreliable link is not restored to the snapshot.
A snapshot at a node~$p$ determines a part of a connected component to which $p$ belongs, but a snapshot may not be up to date, as links may fail and there could be a delay in $p$ learning about it.

We say that \emph{a node~$p$ has heard of a node~$q$} if \texttt{name}$_q$ is a vertex in the snapshot at~$p$.
Each node hears of its neighbors by the end of the first round.
A node~$q$ that belongs to the same connected component as a node~$p$, according to the current snapshot at~$p$, and such that $p$ knows the initial input of~$q$, is considered as \emph{settled by~$p$}.

A node~$p$ executing the algorithm participates in exchanging snapshots with neighbors, as long as its connected component contains nodes that have not been settled yet, according to the current snapshot.
At the end of the first round, every node considers all its neighbors as pending settling.
A soon as a node realizes that all nodes in its connected component according to the snapshot are settled,  it sends its snapshot to the neighbors for the last time, finds the maximum input among the nodes in its snapshot, decides on this maximum value, and halts.


\begin{figure}
	
\hrule
	
\FF
	
\noindent
\texttt{Algorithm} \textsc{ES-Agreement}
	
\FF
	
\hrule
	
\FF
	
\begin{enumerate}[nosep]
\item 
\label{early-stopping-initialization}
initialize: 
$\texttt{Nodes} \leftarrow \{\texttt{name}_p \}$,
$\texttt{Inputs} \leftarrow \{(\texttt{name}_p,\texttt{input}_p) \}$,  
$\texttt{Links} \leftarrow \emptyset$,
$\texttt{Unreliable} \leftarrow \emptyset$
		
\item \label{first-round-communication} 
\texttt{for} each port \texttt{do} 
\begin{enumerate}[nosep]
\item
send \texttt{name}$_p$ through this port
\item 
\texttt{if} \texttt{name}$_q$ received  through this port \texttt{then}

\begin{enumerate}[nosep]
\item
assign \texttt{name}$_q$ to the port as a name of the  neighbor 
\item
add \texttt{name}$_q$ to \texttt{Nodes};  
add edge $\{\texttt{name}_p, \texttt{name}_q \}$ to \texttt{Links}
\end{enumerate}
\end{enumerate}
		
\item \label{while_line}
\texttt{while}  \text{there exists an unsettled node in $p$'s connected component in the snapshot}  \texttt{do}
\begin{enumerate}[nosep]
\item[]
\texttt{for} each neighbor $q$  \texttt{do}
\begin{enumerate}[nosep]
\item
send sets \texttt{Nodes}, \texttt{Links}, \texttt{Unreliable}, \texttt{Inputs} to $q$ 

\item
\texttt{if} a message from $q$  was just received \texttt{then}
\begin{enumerate}[nosep]
\item[]
update the sets \texttt{Nodes}, \texttt{Links}, \texttt{Unreliable}, \texttt{Inputs} 

\ \ \ \ \ \ \ \ \ \ by adding new elements included in this message from $q$
\end{enumerate} 
\item[] 
\texttt{else}  add edge $\{\texttt{name}_p, \texttt{name}_q\}$ to \texttt{Unreliable} 

\end{enumerate}
\end{enumerate}	
\item
\label{one-more-time}
\texttt{for} each neighbor $q$  \texttt{do}
send sets \texttt{Nodes}, \texttt{Links}, \texttt{Unreliable}, \texttt{Inputs} to $q$

\item \label{decision_line} 
decide on the maximum input value at the second coordinate of a pair in \texttt{Inputs}
\end{enumerate}
	
\FF
	
\hrule
	
\FF

\caption{\label{fig:early-stopping-algorithm}
A pseudocode for a node~$p$.
A node~$q$ is considered unsettled by $p$ if it is in the same connected component as $p$, according to the snapshot at $p$, and there is no pair of the form $(\texttt{name}_q, ?)$ in \texttt{Inputs}$_p$.}
\end{figure}

\Paragraph{An implementation of the algorithm.}

The algorithm is called \textsc{ES-Agreement}, its pseudocode is given  in Figure~\ref{fig:early-stopping-algorithm}.
The pseudocode refers to a number of variables that we introduce next.
A set variable \texttt{Nodes} at a node~$p$ stores the names of all the nodes that the node~$p$ has  ever learned about, and a set variable \texttt{Links} stores the links known by~$p$ to have transmitted messages successfully at least once, a link is represented as a set of two names of nodes at the endpoints of the link. 
A set variable \texttt{Unreliable} stores the edges representing links known to have failed. 
Knowledge about failures can be acquired in two ways: either directly, when a neighbor is expected to send a message at a round and no message arrives through the link, or indirectly, contained in a snapshot received from a neighbor.
A node stores all known initial input values of nodes $q$ as pairs $(\texttt{name}_q, \texttt{input}_q)$  in a set variable \texttt{Inputs}.
The nodes keep notifying their neighbors of the values of some of their private variables during iterations of the while loop in instruction~\eqref{while_line} in Figure~\ref{fig:early-stopping-algorithm}.
A node iterates this loop until all vertices in the connected component of the node are settled, which is sufficient to decide.
Once a node is ready to decide, it forwards its snapshot to all the neighbors for the last time, decides on the maximum input value in some pair in \texttt{Inputs}, and halts.

An execution of the algorithm starts with each node announcing  its name to all its neighbors, by executing the instruction~\eqref{first-round-communication} in Figure~\ref{fig:early-stopping-algorithm}.
This allows every node to discover its neighbors and map its ports to the neighbors' names.
A node does not send its input in the first round of communication.
A node sends its snapshot to the neighbors for the first time at the second round, by  instruction~\eqref{while_line} in the pseudocode in Figure~\ref{fig:early-stopping-algorithm}.

A node~$p$ has heard of a node~$q$ if \texttt{name}$_q$ is in the set \texttt{Nodes}$_p$.
A node~$p$ has settled node~$q$ once the pair $(\texttt{name}_q, \texttt{input}_q)$ is in \texttt{Inputs}$_p$ and the node~$q$ belongs to the connected component of~$p$ according to its snapshot.
We say that a \emph{node~$p$ knows the state of a node~$q$ at the end of a round~$i$} if the following inclusions hold: $\texttt{Nodes}_q\subseteq \texttt{Nodes}_p$, $\texttt{Links}_q\subseteq \texttt{Links}_p$, $\texttt{Unreliable}_q\subseteq \texttt{Unreliable}_p$, and $\texttt{Inputs}_q\subseteq \texttt{Inputs}_p$, where \texttt{Nodes}$_q$, \texttt{Links}$_q$, \texttt{Unreliable}$_q$, and \texttt{Inputs}$_q$ denote the values of these variables at~$q$ at the end of round~$i$.
If the node~$p$ hears of its neighbor~$q$ at the first round, then $p$ knows only the $q$'s name, but does not know either the input or any neighbor of~$q$ other than oneself.

\Paragraph{The correctness and performance.}

We show that the algorithm is a correct disconnected agreement solution that is early stopping.

\begin{lemma}
\label{lem:know-state-after-round}

Once a node~$p$ settles a node~$q$, then $p$ knows the state of $q$ at the end of the first round.
\end{lemma}

\begin{proof}
After the first round, the set \texttt{Inputs} at~$q$ includes the only pair $(\texttt{name}_q,\texttt{input}_q)$.
This is because of the initialization in instruction~\eqref{early-stopping-initialization} of the pseudocode in Figure~\ref{fig:early-stopping-algorithm}, 
and since $q$ does not receive the neighbors' inputs at the first round, by instructions~\eqref{first-round-communication}.
Node $q$ learns the names of its neighbors during the first round, which $q$ uses to populate  \texttt{Nodes} and \texttt{Links}.
During the first iteration of the while loop in Figure~\ref{fig:early-stopping-algorithm}, which occurs at the second round, the node~$q$ sends the pair $(\texttt{name}_q,\texttt{input}_q)$ to the neighbors, along with the contents of sets \texttt{Nodes}, \texttt{Links}, and \texttt{Unreliable}, as they were at the end of the first round.
Pair $(\texttt{name}_q,\texttt{input}_q)$ spreads through the network carried in snapshots sent to neighbors, along with the contents of the sets \texttt{Nodes}, \texttt{Links}, and \texttt{Unreliable} at node~$q$.
When a pair $(\texttt{name}_q,\texttt{input}_q)$ reaches a node~$p$ for the first time, the contributions of the sets \texttt{Inputs}$_q$, \texttt{Nodes}$_q$, \texttt{Links}$_q$, and \texttt{Unreliable}$_q$ to the received snapshot are as from these set variables at the end of the first round at~$q$.
\end{proof}

The nodes settle their neighbors by the end of the second round, after the first iteration of the while loop  in Figure~\ref{fig:early-stopping-algorithm}, as long as the links to these neighbors have not failed.
As the while loop iterates in consecutive rounds, once a node~$p$ hears of some node~$q$ at a round~$i$, then the earliest $p$ is expected to settle the node~$q$ is at the next round~$i+1$.
To see this, observe that a node~$p$ hears of its neighbors at the first round and settles them in the second round, unless some links to neighbors failed.
A failure of a link to a neighbor may postpone settling this neighbor, if its input value manages to reach $p$ eventually, or $p$ may never settle this neighbor, if it gets disconnected from~$p$. 
The delay of at least one round between hearing of a node and settling this node is maintained through each iteration of the while loop.
Once node~$p$ settles a node~$q$, it learns the state of~$q$ at the end of the first round, by Lemma~\ref{lem:know-state-after-round}.
Such a state may include the names of the nodes in \texttt{Nodes}$_q$  that are unsettled by~$p$ yet; if this is the case then the node~$p$ continues iterating the while loop.

After each round, a node builds a snapshot as an approximation of the network's topology.
The vertices of this graph are the names of nodes from the set \texttt{Nodes}, and these  links in the set \texttt{Links} that are not in \texttt{Unreliable} make the edges.
We say that a node~$p$ \emph{completes survey} of the network by a round if $p$ has settled all nodes in its connected component according to the snapshot of this round.
A node keeps communicating with neighbors until it completes survey, and then one more time, by instruction~\eqref{one-more-time} in the pseudocode in Figure~\ref{fig:early-stopping-algorithm}.
This extra round of communication serves the purpose to help the neighbors complete their surveys in turn, as they may need the information that has just allowed the sender to complete its survey.
Finally, a node decides on the maximum from the set of all the input values stored in the pairs in \texttt{Inputs}, and halts.

\begin{lemma}
\label{lem:early-stop-agreement}

If a node~$p$ decides on \texttt{input}$_r$, another node~$q$ also decides, and the nodes~$p$ and~$q$ are connected by a reliable path at a round when they have already decided, then $q$ has the node~$r$  as settled in its snapshot at this round.
\end{lemma}

\begin{proof}
Suppose $q$ does not have $r$ in its snapshot as settled at the first round in which both $p$ and~$q$ have completed surveys, to arrive at a contradiction.
The node~$p$ has $q$ in its connected component of the snapshot and, similarly, the node~$q$ has $p$ in its connected component of the snapshot, because they are connected by a reliable path at the first round after completing surveys.
Let $\gamma=(s_1,\ldots, s_k)$ be a reliable path connecting $p=s_1$ with $q=s_k$ at the first round in which both $p$ and $q$ have completed surveys and such that $q$ settled~$p$ by a snapshot that arrived through this path.
Let a node~$s_i$ on the  path~$\gamma$ be such that $i$ is the greatest index $j$ of a node~$s_j$ in $\gamma$ such that $s_j$  settled~$r$ in its snapshot.
An index~$i$ with this property exists because node~$p=s_1$ is such.
Moreover, the inequality $i<k$ holds because $q=s_k$ does not have $r$ settled in its snapshot.
There is a reliable path $\delta=(t_1,\ldots,t_m)$ from~$r=t_1$ to~$s_i=t_m$ through which a snapshot arrived first bringing  $\texttt{input}_r$ to make $s_i$ settle~$r$.
Consider a path $\zeta$ obtained by concatenating $\delta$ with a part of $\gamma$ starting at $s_i$ and ending at~$s_k=q$, where $i<k$.
Let us denote the nodes on this path as $(u_1,\ldots,u_\ell)=\zeta$, where $u_1=r$ and $u_\ell=q$, for reference.

We examine the flow of information along~$\zeta$ from~$r=u_1$ towards~$u_\ell=q$.
At the first round, the node~$u_\ell=q$ learns the name of its neighbor~$u_{\ell-1}$.
At the second round, node~$u_\ell$ settles~$u_{\ell-1}$ and learns of the node~$u_{\ell-2}$, by Lemma~\ref{lem:know-state-after-round}.
In general, a node~$u_j$, such that $j>1$, hears of its neighbor~$u_{j-1}$ at the first round and settles it at the second round.
At a round a node~$u_j$ settles its neighbor~$u_{j-1}$, it also hears of the node~$u_{j-2}$ as still unsettled. 
This creates  a chain of dependencies such that the node~$u_j$ heard of a node up the path $\zeta$ towards~$r$ that is still unsettled and in the same connected component in its snapshot.

As snapshots with \texttt{input}$_r$ move along $\zeta$ towards~$q$, this chain of dependencies, starting at a node that received \texttt{input}$_r$ most recently and ending at~$q$, stays unbroken.
This is because of the following two reasons. 
First, the part of $\zeta$ taken from~$\delta$ provides reliable edges at all times, since \texttt{input}$_r$ manages to reach~$s_i$: the only possibility of  this not being the case would be to settle a node on this path via a different shorter path to $s_i$, but this is a shortest path by its choice. 
Second, the part consisting of $\gamma$ provides reliable edges during the considered rounds, since these edges are still reliable when $q$ completes survey. 
This makes node~$q$ eventually hear of~$r$, and receive \texttt{input}$_r$ at the next round, which is a contradiction.
\end{proof}


\begin{theorem}
\label{thm:ES-Agreement}

Algorithm \textsc{ES-Agreement} is an early stopping solution of disconnected agreement  that relies on minimal knowledge, 
terminates within $\lambda+2$ rounds 
and uses messages carrying $\cO(m\log n)$ bits.
\end{theorem}

\begin{proof} 
A node decides on an input value from its snapshot, which gives validity.

We show agreement as follows. 
Consider two nodes $p$ and $q$ that are connected by a reliable path at the first round when each of these nodes has already decided.
Suppose, to arrive at a contradiction, that node~$p$ decided on a value that is greater than the value that node~$q$ decided on.
Let $r$ be the node that provided its \texttt{input}$_r$ as the decision value for~$p$.
By Lemma~\ref{lem:early-stop-agreement}, node~$q$ has node~$r$ settled at the round of deciding, so $q$ decides on a value that is at least as large as \texttt{input}$_r$, which is a contradiction.

Next, we estimate the number of rounds needed for each node to halt.
As an execution proceeds, information flows through the connected components of the network, by iteratively sending a snapshot to the neighbors and updating it at the same round.
A value that gets decided on, in a particular connected component, may travel along a path that shares its parts with multiple connected components. 
If such a path crosses a connected component then its length is upper bounded by the connected component's diameter.
If a link on such a path fails, this may occur after the future decision value got transmitted through this link, and the endpoints of this link may belong to different connected components.
This shows that the number of hops a decision value makes on its way between a pair of nodes it at most the stretch.
There are only two rounds that cannot be accounted for by this counting: the first round, during which the nodes discover their neighbors, and the last round,  when a node notifies its neighbors of its snapshot for the last time.
So if an execution terminates at a round~$t$, then the stretch at this round is at least $t-2$. 
It follows that the algorithm terminates by the round $\lambda+2$. 
\end{proof}

\section{Optimizing Link Use}

\label{sec:link-use}

We present an algorithm solving disconnected agreement that uses the optimal number $\cO(n)$ of links and messages of $\cO(m\log n)$ bits.
Optimizing the link use in order to achieve an algorithmic goal could be interpreted as relying on a network backbone to accomplish the task and building such a backbone on the fly. 
We depart from the model of minimal knowledge of the previous sections and assume that nodes know their neighbors at the outset, in having names of the corresponding neighbors associated with all their ports.
We identify links, determined by the ports of a node, by the names of the respective neighbors of the node, and use the terms ports and incident links interchangeably.
The disconnected agreement algorithm we describe next makes nodes halt in $\cO(n m)$ rounds. 
We complement the algorithm by showing that using $\cO(n)$ links is only possible when each node starts with a mapping of ports on its neighbors, because otherwise $\Omega(m)$ is a lower bound on the link use.
We also show that no algorithm can simultaneously be early stopping and use $\cO(n)$ links.

\Paragraph{An overview of the algorithm.}

The general idea of the algorithm is to have nodes build their maps of the network, representing the topology, that include the connected component of each node.
In this the algorithm resembles \textsc{ES-Agreement}.
An approximation of the map at a node evolves through a sequence of snapshots of the vicinity of the node.
Such a snapshot helps to coordinate choosing links through which messages are sent to extend the current snapshot to a bigger one.
We want to accomplish manipulating snapshots by sending messages through as few links as possible, meaning $\cO(n)$ links.
This is possible in principle, because a spanning forest has $\cO(n)$ links, and at the start each node already knows its neighbors.
Input values could be a part of node attributes of the vertices on such a map.
After the process of drawing a map is completed, which includes identifying a connected component, a decision can be made based on the information included in the map.

A node categorizes its incident links as either passive, active or unreliable; these are exclusive categories that evolve in time. 
An \emph{active} link is used to send messages through it, so a node categorizes an incident link as active once it receives a message through it.
Initially, one link incident to a node is considered as active by the node, and all the remaining incident links are considered passive.
A link is \emph{passive} at a round if none of its endpoint nodes has ever attempted a transmission through this link.
A node transmits through an active port at every round, unless the node decides and halts.
It follows that if a node~$p$ considers a link active, which connects it to a neighbor~$q$, then $q$ considers the link active as well, possibly with a delay of one round.
Similarly,  if a node~$p$ considers a link passive, which connects it to a neighbor~$q$, then $q$ considers the link passive as well, possibly for one round longer than~$p$.
A node~$p$ detects a failure of an active link and begins to consider it unreliable after the link fails to deliver a  message to~$p$ as it should.
For an active link connecting a node $p$ with $q$, once $p$ considers the link unreliable then $q$ considers the link unreliable as well, possibly with a delay of one round.

The \emph{state of a node~$p$ at a round} consists of its name, the input value, and a set of its neighbors, with each incident link categorized as either passive, active, or unreliable, representing this categorization of links by the node~$p$ at the round.
The states of a node may evolve in time, in that an incident link may change its categorization.
Links start as passive, except for one incident link per node initialized as active, then they  may become active, and finally they may become unreliable. 

A snapshot of the network at a node represents the node's knowledge of its connected component in the network restricted to the active edges and the  states of its nodes.
Formally, a \emph{snapshot of network} at a node~$p$ at a round is a collection of states of some nodes that $p$ has received and stores.
A snapshot allows to create a map of a portion of the network, which is a graph with the names of nodes as vertices and the edges representing links. 
This map can include the input values of some nodes, should they become known.
A connected component of a node with other nodes reachable by active links is a part of such a map. 
Formally, the \emph{active connected component} of a node~$p$ at a round is a connected component, of the vertex representing $p$, in a graph that is a map of the network according to the snapshot of $p$ at the round with only active links represented by edges.

A node~$p$ sends a summary of its knowledge of the states of nodes in the network to the neighbors through all its active links at each round. 
If $p$ receives a message with such knowledge from a neighbor, then $p$ updates its knowledge and the snapshot by incorporating the newly learned information. 
Such new information may include either a state of a node~$q$, such that $p$ has never had a  state of~$q$, or a subsequent state of node~$q$, such that $p$ has already had some state of~$q$. 
At each round, a node~$p$ determines its active connected component based on the current snapshot.
If we refer to an active connected component of~$p$ then this means the active connected component according to the current snapshot.
We say that \emph{a node~$p$ has heard of a node~$q$} if the \texttt{name}$_q$ occurs in  the snapshot at~$p$; the node~$p$ may either store some $q$'s state or $q$'s name may belong to a state of some other node that $p$ stores.
A node~$p$ considers another node~$q$ \emph{settled} if $p$ has $q$'s state in its snapshot.
A node~$p$ considers its active connected component \emph{settled} if $p$ has settled all the nodes in its active connected component.

If a node~$p$ has heard about another node~$q$ such that $q$ does not belong to the node~$p$'s active connected component, but it is connected to a node~$r$ in the active connected component by a passive link, then the node~$p$ considers the link connecting $q$ to~$r$ as \emph{outgoing}.
If there is an outgoing link in $p$'s active connected component then $p$ considers its active connected component \emph{extendible}, otherwise $p$ considers its active connected component \emph{enclosed}.

We want the nodes to participate in making some outgoing links active, each time the active connected component is settled and still extendible.
All the nodes in the active connected component of a node~$p$ can choose the same outgoing link to make it active, once the active connected component becomes settled, because  each node knows the same set of outgoing links.
Once $p$'s active connected component becomes settled and enclosed then $p$ may decide.


\begin{figure}[t]
\hrule
\FF
\texttt{algorithm} \textsc{OL-Agreement} 
\FF
\hrule
\FF
\begin{enumerate}[nosep]
\item \label{OL:initialization}	
initialize:
$\texttt{Unreliable} \leftarrow \emptyset$, 
$\texttt{Active} \leftarrow \{ \{ p,q\}\}$ where $q$ is some neighbor, 

$\texttt{Passive} \leftarrow$ set of links to $p$'s neighbors, except for the neighbor $q$ used in  \texttt{Active},

$\texttt{state} \leftarrow (\texttt{name}_p,\texttt{input}_p, \texttt{Active},\texttt{Passive},\texttt{Unreliable})$, 

$\texttt{round} \leftarrow 0$, 
$\texttt{timestamp} \leftarrow(\texttt{state},\texttt{round})$

\item  \label{OL:outer-repeat-loop}
\texttt{repeat}
\begin{enumerate}[nosep]
\item \label{OL:start-epoch}
$\texttt{epoch} \leftarrow \texttt{round}$, 
$\texttt{Snapshot} \leftarrow \{\texttt{state} \}$
\item 
\texttt{repeat}  \label{OL:inner-repeat-loop}
\begin{enumerate}[nosep]
\item \label{OL:round-increment}
$\texttt{round} \leftarrow \texttt{round} + 1$,
add \texttt{timestamp} to set \texttt{Timestamps} 
\item \label{OL:for-each-port-do}
\texttt{for} each incident link $\alpha$ \texttt{do}
\begin{enumerate}[nosep]
\item
\texttt{if} $\alpha$ is in \texttt{Active} \texttt{then} send \texttt{Timestamps}  through $\alpha$
\item 
\texttt{if} $\alpha$ is mature in \texttt{Active} and no message received through $\alpha$ 
\begin{itemize}[nosep]
\item[]
\texttt{then} move $\alpha$ to \texttt{Unreliable}
\end{itemize}
\item 
\texttt{if} a message received through $\alpha$ \texttt{then} place $\alpha$ in  \texttt{Active}
\end{enumerate}
\item 
\texttt{for} each received timestamp pair  $(\texttt{state},y)$ \texttt{do}
\begin{enumerate}[nosep]
\item
add $(\texttt{state},y)$ to \texttt{Timestamps}
\item \label{OL:add-state-snapshot}
\texttt{if} $y> \texttt{epoch}$ \texttt{then} add $\texttt{state}$ to \texttt{Snapshot}
\end{enumerate}
\end{enumerate}	
\item \label{OL:inner-until-condition}
\texttt{until} the active connected component is settled 
\item \label{OL:add-connector}
if the active connected component is extendible then
\begin{enumerate}[nosep]
\item
identify an outgoing edge as a connector
\item 
if the connector is incident to $p$ then place it in \texttt{Active}
\end{enumerate}
\end{enumerate}
\item \label{OL:outer-until-condition}
\texttt{until} the active connected component is enclosed
\item  \label{OL:decision-value}
set $\texttt{candidate}_p$ to the maximum input value in \texttt{Snapshot}
\item \label{OL:notify-decision}
send pair $(\textrm{this-is-decision},\texttt{candidate}_p)$ through each active incident link
\item \label{OL:decision} 
decide on \texttt{candidate}$_p$ 
\end{enumerate}
\FF
\hrule
\caption{\label{fig:OL-agreement}
A pseudocode for a node~$p$.
In each round, node~$p$ checks to see if a pair of the form $(\texttt{decision},z)$ has been  received, and if so then $p$ forwards this pair through each active port, decides on~$z$, and halts. }
 \end{figure}

\Paragraph{An implementation of the algorithm.}

The algorithm is called \textsc{OL-Agreement}, its pseudocode  is in Figure~\ref{fig:OL-agreement}.
Each node stores links it knows as unreliable in a set \texttt{Unreliable}, initialized to the empty set. 
Each node stores links it considers active in a set \texttt{Active}, initialized to some incident link. 
Each node  stores passive links in a set \texttt{Passive}, which a node initializes to the set of all incident links except for the one initially activated link. 

All nodes maintain a variable \texttt{round} as a counter of rounds.
In each round, a node creates a \emph{timestamp pair}, which consists of its current state and the value of the round counter used as a timestamp.
A node~$p$ stores timestamp pairs in a set \texttt{Timestamps}. 
For each node~$q$ different from~$p$, a node~$p$ stores a timestamp pair for~$q$ if such a pair arrived in messages and only one pair with the largest timestamp.
These variables are initialized by instruction~\eqref{OL:initialization} in Figure~\ref{fig:OL-agreement}.

The initialization is followed by iterating a loop performed by instruction~\eqref{OL:outer-repeat-loop} in the pseudocode in Figure~\ref{fig:OL-agreement}.
The purpose of an iteration is to identify a new settled active connected component; we call an iteration \emph{epoch}.
An epoch is determined by the round in which it started, remembered in the variable~\texttt{epoch} by instruction~\eqref{OL:start-epoch}.
The knowledge of an active connected component of a node~$p$ identified in an epoch is stored in a set \texttt{Snapshot}, which is initialized at the outset of an epoch to the $p$'s state by instruction~\eqref{OL:start-epoch}.
This knowledge is represented as a collection of states of nodes that arrived to $p$ in timestamp pairs, with timestamps indicating that they were created after the start of the current epoch, as verified by instruction~\eqref{OL:add-state-snapshot}.
The main part of an epoch is implemented as an inner repeat loop~\eqref{OL:inner-repeat-loop}.
An iteration of this loop implements a round of communication with neighbors through active links and updating the state by instruction~\eqref{OL:for-each-port-do}.
An incident link in \texttt{Active} is \emph{mature} if either it became active because a message arrived through it or $p$ made it active spontaneously at some round $i$ and the current round is at least $i+2$.
If a mature active link fails to deliver a message then $p$ moves it to \texttt{Unreliable}.

A set variable \texttt{Timestamps} stores timestamp pairs that a node sends in each message and updates after receiving messages at a round.
A set variable \texttt{Snapshot} is used to construct an active connected component.
\texttt{Snapshot} is rebuilt in each epoch, starting only with the current $p$'s state.
We separate storing timestamp pairs in a set \texttt{Timestamps} used for communication from storing states in \texttt{Snapshot} to build an active connected component, to facilitate a proper advancement of epochs in other nodes.

We say that node~$p$ \emph{completes the survey} of the network by a round if $p$ has settled all the nodes in its active connected component according to the snapshot of this round.
The inner repeat loop~\eqref{OL:inner-repeat-loop} terminates once $p$ completes the survey of the network, by condition~\eqref{OL:inner-until-condition} controlling the loop.
If the active connected component is extendible, then $p$ identifies a \emph{connector} which is an outgoing edge to be made active.
We may identify an outgoing edge that is minimal with respect to the lexicographic order among all the outgoing links for a settled active connected component to be designated as a connector.
If a connector is a link incident to $p$ then $p$ moves it to the set \texttt{Active}, by instruction~\eqref{OL:add-connector}.

Once an epoch ends, by condition~\eqref{OL:inner-until-condition}, and the active connected component is enclosed, then the main repeat loop ends, by condition~\eqref{OL:outer-until-condition} controlling the loop.
At this point, a node~$p$ is ready to decide, and the decision is on the maximum input value in a state stored in \texttt{Snapshot}, by instruction~\eqref{OL:decision-value}.
Node $p$ notifies each neighbor connected by an active link of the decision value, by instruction~\eqref{OL:notify-decision}, and decides, by instruction~\eqref{OL:decision}.
The pseudocode in Figure~\ref{fig:OL-agreement} omits what pertains to handling messages that could be generated in round~\eqref{OL:notify-decision} by some nodes.
Namely, as a first thing at a round, a node~$p$ verifies if a pair of the form $(\texttt{decision},z)$ has been  received in the previous round, and if so then $p$ forwards this pair through all active ports, decides on this value~$z$, and halts.

\Paragraph{The correctness and performance bounds.}

We show that algorithm \textsc{OL-Agreement} is a correct solution to disconnected agreement, and estimate its performance.
The algorithm uses a similar paradigm to survey connected components as algorithm \textsc{ES-Agreement}.
A node first learns of other nodes' names and later of these nodes' input values.
This creates a chain of dependencies along paths through connected components, which in the case of algorithm \textsc{OL-Agreement} consist of active links.

\begin{lemma}
\label{lem:OL-components}

If two nodes $p$ and $q$ are connected by a reliable path just after they both decided then each of them counts the other node in its last active connected component.
\end{lemma}

\begin{proof}
Let us consider an arbitrary reliable path from $p$ to~$q$.
Each link on this path is either active or passive.
This is because once a passive link becomes active, then a message is sent trough it in each round, so a failure is immediately detected.
A passive link is never activated only if it already connects two nodes in the same active connected component.
An active link makes its endpoints belong to the same active connected component.
\end{proof}

The following Lemma~\ref{lem:OL-agreement} is analogous to Lemma~\ref{lem:early-stop-agreement} about algorithm \textsc{ES-Agreement}, which also uses sufficiently large messages to build a map approximating the network.
We need it to show agreement.
A proof could be structured similarly to that of Lemma~\ref{lem:early-stop-agreement}; we include a detailed argument for the sake of completeness.

\begin{lemma}
\label{lem:OL-agreement}

If a node~$p$ decides on \texttt{input}$_r$, another node~$q$ also decides, and nodes~$p$ and~$q$ are connected by a reliable path at a round when they have already decided, then $q$ has node~$r$ as settled in its snapshot at this round.
\end{lemma}

\begin{proof}
Suppose $q$ does not have $r$ in its snapshot as settled at the first round in which both $p$ and~$q$ have completed their surveys, to arrive at a contradiction.
The nodes $p$ and $q$ have each other in their active connected components, by Lemma~\ref{lem:OL-components}.
Let $\gamma=(s_1,\ldots, s_k)$ be a path consisting of the active links that connect $p=s_1$ with $q=s_k$ just after both $p$ and $q$ have completed their surveys and such that $q$ settled~$p$ by a chain of \texttt{Timestamps} that arrived through this path.
Let a node~$s_i$ on the  path~$\gamma$ be such that $i$ is the greatest index $j$ of a node~$s_j$ in~$\gamma$ that has settled~$r$ in its snapshot.
Such an index~$i$ exists because the node~$p=s_1$ has this property.
The inequality $i<k$ holds because $q=s_k$ does not have $r$ settled in its snapshot.
There is a path $\delta=(t_1,\ldots,t_m)$ from~$r=t_1$ to~$s_i=t_m$ consisting of active links through which \texttt{Timestamps} arrived first bringing  $\texttt{state}_r$ to enable $s_i$ to settle~$r$.
Consider a path $\zeta$ consisting of the active links obtained by concatenating $\delta$ with a part of $\gamma$ starting at $s_i$ and ending at~$s_k=q$, where $i<k$.
Let us denote the nodes on this path as $(u_1,\ldots,u_\ell)=\zeta$, where $u_1=r$ and $u_\ell=q$.

At the first round, the node~$u_\ell=q$ learns a state of its neighbor~$u_{\ell-1}$.
At the second round, the node~$u_\ell$ settles~$u_{\ell-1}$ and hears of the node~$u_{\ell-2}$.
In general, a node~$u_j$, such that $j>1$, hears of its neighbor~$u_{j-1}$ at the first round and settles it at the second round.
If a node~$u_j$ settles its neighbor~$u_{j-1}$ at a round then it also hears of the node~$u_{j-2}$ as still unsettled. 
This creates  a chain of dependencies such that the node~$u_j$ heard of a node up the path~$\zeta$ towards~$r$ that is still unsettled and in the same connected component in its snapshot.

As \texttt{Timestamps} with \texttt{state}$_r$ move along $\zeta$ towards~$q$, this chain of dependencies, starting at a node that received \texttt{state}$_r$ most recently and ending at~$q$, stays unbroken.
This is because of the following two reasons. 
First, the part of $\zeta$ taken from~$\delta$ provides active edges at all times, since \texttt{state}$_r$ manages to reach~$s_i$: the only possibility of  this not being the case would be to settle a node on this path via a different shorter path to~$s_i$, but this is a shortest path by its choice. 
Second, the part consisting of $\gamma$ provides active edges during the considered rounds, since these edges are still reliable when $q$ completes the survey. 
This makes node~$q$ eventually hear of~$r$, and receive \texttt{state}$_r$ at the next round, which is a contradiction.
\end{proof}

We consider an auxiliary \emph{activation process} on a graph that models the activation of connectors in an execution of algorithm \textsc{OL-Agreement}.
The process acts on a given simple graph $H$ that has some~$k$ vertices and proceeds through consecutive rounds.
An edge of $H$ may progress through three states, first passive, then possibly active, and then possibly be deleted.
We consider subgraphs determined by active edges, so  connected components are active connected components.
An edge connecting a vertex in a connected component $C$ to a vertex in a different connected component is considered as outgoing from~$C$.
In the beginning of a round, if an active connected component has an outgoing edge, then one such an edge is made active.
At the end of a round, some active edges may be deleted.
The process continues until there are no outgoing edges.

\begin{lemma}
\label{lem:OL-activation-process}

In the course of the activation process on a graph with $k$ vertices, the number of active edges in the graph at any round is at most $2k-2$.
\end{lemma}

\begin{proof}
We start with no active edges, so each vertex is an active connected component and $k$ is the number of such components.
Consider the activation of new edges at a round as occurring sequentially.
As an edge is activated, it may connect two different active connected components, thus decreasing their number, or it may close a cycle.
Suppose an active connected component~$C_1$ gets connected to an active connected component~$C_2$, then $C_2$ to $C_3$, and so on through~$C_i$, with the edge activated in~$C_i$ connecting $C_i$ to some~$C_j$, for $1\le j< i$.
That last edge from a vertex in~$C_i$ to a vertex in~$C_j$ is not needed to make all~$C_i$, for $1\le j\le i$, into one connected component, so we treat it as an \emph{extra} edge.
A number of such extra edges created at a round is maximized when $i=2$, because  for each decrease of the number of connected components with one activated edge we also activate another edge.
A tree minimizes the number of connected components to one, and a tree on $k$ vertices has $k-1$ edges.
We obtain that the number of active edges at each round is at most twice the number of edges in a tree of $k$ vertices, which is $2(k-1)$.
\end{proof}

\begin{theorem}
\label{thm:OL-Agreement}

Algorithm \textsc{OL-Agreement} solves disconnected agreement  in $\cO(n m)$ rounds with fewer than $2n$  links used at any round and sending messages of $\cO(m\log n)$ bits. 
\end{theorem}

\begin{proof} 
A node decides on an input value from a state in its snapshot.
All states include original input values, which gives validity.

Consider two nodes $p$ and $q$ connected by a reliable path at the first round when each of these nodes has already decided.
We want to show that $p$ and $q$ decide on the same value.
Suppose, to arrive at a contradiction, that the node~$p$ has decided on a value that is greater than the value that the node~$q$ has decided on.
Let $r$ be the node whose state provided its input value as the decision value for~$p$.
By Lemma~\ref{lem:OL-agreement}, the node~$q$ has the node~$r$ settled at the round of deciding, so $q$ decides on a value that is at least as large as \texttt{input}$_r$, which is a contradiction.
This gives agreement.

As an execution proceeds, information flows through each active connected component, by iteratively sending timestamp pairs to neighbors via active links and simultaneously updating the latest states in timestamp pairs.
The length of a path of active links traversed by a timestamp pair may be as long as the number of nodes in an active connected component minus one.
This means that an epoch takes fewer than~$n$ rounds.
After a new connector is added, it takes the length of an epoch for all the nodes to settle on the states of the nodes in an active connected component.
There may be up to $m$ links added as connectors.
This means that every node halts in $\cO(n m)$ rounds.

The process of making links active could be modeled as the activation process on a graph representing the network.
By Lemma~\ref{lem:OL-activation-process},  the number of links that are active at the same time is less than~$2n$.
\end{proof}

\Paragraph{Lower bounds for link usage.}

We now consider a setting in which the destinations of ports are not initially known to nodes.
For any positive integers $n$ and $m$ such that $m=\cO(n^2)$, we design a graph $\cG(n,m)$ with $\Theta(n)$ vertices and $\Theta(m)$ edges, which makes any disconnected agreement solution to use $\Theta(m)$ links even if the nodes know the parameters $n$ and $m$.
We drop the parameters $n$ and $m$ from the notation $\cG(n,m)$, whenever they are fixed and understood from context, and simply use~$\cG$. 

Consider any positive integers $n$ and $m$ such that $m=\cO(n^2)$.
Let graph~$\cG$ consist of two identical parts $G_{1}$ and~$G_{2}$ as its subgraphs. 
The parts are $\left \lceil{\frac{m}{n}}\right \rceil$-regular graphs of $\left \lceil{\frac{n}{2}}\right \rceil$ vertices each. 
Without loss of generality, we can assume that the number $\left \lceil{\frac{m}{n}}\right \rceil$ is even, to guarantee that such regular graphs exist. 
Graph~$\cG$ is obtained by connecting $G_1$ and $G_2$ with $\left \lceil{\frac{n}{2}}\right \rceil$ edges such that each vertex from~$G_{1}$ has exactly one neighbor in~$G_{2}$.

By the construction, graph~$\cG$ has $2 \left \lceil{\frac{n}{2}}\right \rceil = \Theta(n)$ vertices and $(\left \lceil{\frac{m}{n}}\right \rceil + 1) \left \lceil{\frac{n}{2}}\right \rceil  = \Theta(m)$ edges.
Let us assume now that the destinations of outgoing links are not initially known to the nodes. 
This means that ports can be associated with neighbors's names only after receiving messages through them.
The following Theorem~\ref{thm:link-use} holds even if  $n$ and $m$ can be a part of code.


\begin{theorem}
\label{thm:link-use}

For any disconnected agreement algorithm $\cA$ relying on minimal knowledge, and positive integer numbers $n$ and $m$ such that $n \le m$ and $m\le n^2$,  there exists a network $\cG(n,m)$ with $\Theta(n)$ nodes and $\Theta(m)$ links and an execution of algorithm $\cA$ on  $\cG(n,m)$ that uses $\Theta(m)$ links. 
\end{theorem}

\begin{proof} 
If a node~$p$ executing $\cA$ sends a message by an incident link, through which it has not received nor sent any messages yet, then the receiving node could be any original neighbor of the node~$p$ in the network~$\cG=\cG(n,m)$ which has not communicated with $p$ yet.
We refer as \emph{uncovered neighbors} of~$p$ at a round to these among $p$'s neighbors that $p$ has not received a message from nor sent a direct message to yet. 
Choosing a node that is an uncovered $p$'s neighbor, identifiable only by its port at~$p$, is a part of the  adversarial strategy of failing links.

Algorithm $\cA$ gets executed on some initial configurations of the network~$\cG$, as defined above. 
Let the notation~$\cC_{i,j}$ mean an initial configuration in which the nodes from the subgraph~$G_{1}$ have $i \in \{0, 1\}$ as their input values while the nodes from the subgraph~$G_{2}$ have $j \in \{0, 1\}$.
We call a link between two nodes in $G_1$ {\em unused}, by a given round, if none of its endpoints has tried to send a message through it yet.
Let $k_{1}$ denote the number of nodes in~$G_{1}$ that have at least one unused link incident to some other node in~$G_2$. 
At the beginning, $k_{1} = |G_{1}|$ and $k_1$ may decrease in the course of an execution.

For any initial configuration $\cC_{i,j}$, where $i,j\in\{0,1\}$, consider the following adversarial strategy to fail links. 
Let a node~$p$ be in~$G_{1}$ and let $\ell$ denote a link that $p$ wants to send a message through at a round.
Suppose node~$p$ has at least two unused links but wants to send a message by only one of them.
The adversary allows the message to be delivered and the other endpoint of~$\ell$ is chosen to be an arbitrary uncovered neighbor of $p$ in~$G_{1}$. 
Suppose there is only one unused link incident to~$p$ and the node~$p$ wants to send a message through it.
If $k_{1} > 1$ then let  this link~$\ell$ fail  before it delivers a message, and otherwise, if $k_{1} = 1$, then let link~$\ell$ deliver the message  to a remaining unassigned neighbor of~$p$. 
It follows that in this case the message must go outside~$G_{1}$. 
Each time number~$k_{1}$ decreases by one, then all but one reliable links incident to some node of $G_{1}$ have been used  and they will not be failed by the adversary in the continuation of the execution. 
Hence, before the first message is delivered from some node in part~$G_{1}$ to some node in part~$G_{2}$, at least $(\left \lceil{\frac{m}{n}}\right \rceil + 1) \cdot (|G_{1}| -1)$ reliable links will have been used for communication between the nodes in~$G_1$.
A similar reasoning applies to communication within part~$G_{2}$ of~$\cG$.

Assume now, to arrive at a contradiction, that in all executions starting from the configurations $\cC_{0,0}$, $\cC_{0,1}$, and $\cC_{1,1}$, in which the above strategy has been used, fewer than $(\left \lceil{\frac{m}{n}}\right \rceil + 1) \cdot (|G_{1}| -1)$ non-faulty links have been used. 
We call these executions $\cE_{0}$, $\cE_b$, and $\cE_{1}$, respectively.
Nodes in part $G_{1}$ do not communicate with nodes in~$G_{2}$ in any of these executions. 
For nodes in the part~$G_{1}$, an execution starting from~$\cC_{0,0}$ is indistinguishable from an execution starting from~$\cC_{0,1}$. 
Similarly, for nodes in the part~$G_{2}$,  an executions starting from $\cC_{0,1}$ is indistinguishable from an execution starting from~$\cC_{1,1}$. 
It follows that each node in the part~$G_{1}$ decides on the same value in the execution~$\cE_b$  as in  the execution~$\cE_{0}$, and a decision has to be on~$0$ by validity.
Each node in the part~$G_{2}$ decides on the same value as in the execution~$\cE_{1}$, and a  decision is on~$1$, by validity. 
This is a contradiction with agreement, as both parts $G_{1}$ and $G_{2}$ are connected in $\cE_b$ and have to decide on the same value.
Therefore, in one of the considered executions, more than these many reliable links have to be used:
$( \lceil{\frac{m}{n}} \rceil + 1) \cdot (|G_{1}| -1) 
= 
( \lceil{\frac{m}{n}} \rceil + 1) \cdot (|G_{2}| -1) $, which is $\Theta(m)$.
\end{proof} 


\begin{theorem}
\label{thm:lower-time-link-use}

Let $\cA$ be a disconnected agreement algorithm that uses $\cO(n)$ reliable links concurrently when executed in networks with $n$ nodes.
For all natural numbers $n$ and $\lambda \le n$, there exists a network~$\cG$ with the stretch~$\lambda$ on which some execution of algorithm~$\cA$ takes $\Omega(n)$ rounds.   
\end{theorem}

\begin{proof} 
Let $n$ and $\lambda\le n$ be natural numbers; we assume that both $n$ and $\lambda$ are even to simplify the notations. 
Let $H_{1}$ be a network with $\frac{n}{2}$ nodes and diameter $\frac{\lambda}{2} - 1$. 
Let $H_{2}$ be a copy of $H_{1}$. 
Let us form a network $\cG$ by taking the nodes in $H_1$ and~$H_2$ and adding all possible links between the nodes in~$H_{1}$ and~$H_{2}$. 
The diameter of $\cG$ is at most $2 \cdot (\frac{\lambda}{2} - 1) + 1 = \lambda - 1$. 

The following is an adversarial strategy to fail links. 
For a round~$i$, let $K_{i}$ be a set of all the reliable links between $H_{1}$ and $H_{2}$ by which nodes attempt to send messages at this round.
The adversary fails all the links from $K_{i}$ before any message arrives, as long as $H_{1}$ and $H_{2}$ stay connected. 
Because there are $\frac{n^2}{4}$ links between the nodes in parts $H_{1}$ and $H_{2}$, this guarantees that no two nodes, of which one is in~$H_{1}$ and the other in~$H_{2}$, exchange a message during at least $\Omega(\frac{n^2}{4n}) = \Omega(n)$ rounds. 

Let us consider the following three initial configurations of~$G$. 
In the first configuration~$\cI_{1}$, all the nodes start with input values~$0$. 
In the second configuration~$\cI_{2}$, all the nodes start with inputs~$1$.
In the third configuration~$\cI_{3}$, the nodes from $H_{1}$ start with the input~$0$, while the nodes from~$H_{2}$ start with the input~$1$. 
We consider executions of algorithm~$\cA$ starting with each of the initial configurations~$\cI_{k}$, for $k=1,2,3$, when the adversary applies the strategy described above to fail edges; let $\cE_k$ be the respective execution. 
If all the nodes halt before some pair of nodes such that one is from~$H_{1}$ and the other is from~$H_{2}$ communicate among themselves, then the nodes in~$H_1$ cannot distinguish the execution~$\cE_1$ from~$\cE_3$ and the nodes in~$H_2$ cannot distinguish the execution~$\cE_2$ from~$\cE_3$.
Let us consider decisions by the nodes in execution~$\cE_3$.
The nodes in~$H_1$ decide on~$0$, as they should in~$\cE_1$, by validity, while the nodes in~$H_2$ decide on~$1$, as they should in~$\cE_2$, by validity.
This contradicts the requirement of agreement.
\end{proof} 

We conclude by settling the question if an algorithm can be simultaneously early stopping and optimize link use.

\begin{corollary}
If a disconnected agreement algorithm uses $\cO(n)$ reliable links concurrently at any time, when executed in networks of $n$ nodes, then this algorithm cannot be early stopping.
\end{corollary}

\begin{proof}
Let us consider integers $n$ and $\lambda_n\le n$ such that $\lambda_n=o(n)$.
By Theorem~\ref{thm:lower-time-link-use}, the algorithm works in $\Omega(n)$ rounds.
The algorithm cannot be early stopping, because this would mean running in time $\cO(\lambda_n)=o(n)$.
\end{proof}

\section{Conclusion}

We introduced the problem of disconnected agreement in the model of networks with links prone to failures such that faulty links may omit messages.
Disconnected agreement is a variant of consensus that has the agreement condition reformulated such that nodes have to agree on the same value only if they belong to the same connected component of a network obtained by removing the faulty links.
The number of values that nodes decide on can be as large as the number of connected components.
As far as allowing for different decision values, the problem is similar to $k$-set agreement, in which agreement is relaxed such that nodes may agree on at most~$k$ values, for a positive integer~$k$.
The essential difference between the two problems is that an acceptable number of decision values in $k$-set agreement is known in advance while the disconnected agreement  has the admissible number of decision values enforced by the changes of topology, and in particular if the network stays connected then all the nodes have to decide on the same value. 
An interesting aspect of $k$-set agreement, as compared to consensus, is that it can be solved in time $\lfloor\frac{t}{k}\rfloor+1$ with up to~$t$ node crashes in synchronous message-passing systems, and this time is necessary, see~\cite{ChaudhuriHLT00}, while consensus requires $t+1$ rounds.
The motivation for disconnected agreement is not to relax consensus such that it can be solved faster, but rather to investigate what running time and size of messages is needed to solve a problem defined by (1)~not imposing any restrictions on which links in a network fail and omit messages, and (2)~minimally relaxing the agreement condition such that solutions exist.
We showed such trade-offs between running time and size of messages, and also looked into minimizing the number of used links.
A possible future direction of work concerns more severe link faults, for example such that result in delivering forged messages.

We demonstrated that a dynamic stretch of a network is a meaningful parameter that serves as a yardstick to measure running time efficiency of disconnected agreement algorithms.
The most running-time efficient disconnected-agreement algorithms are early stopping in that they work in~$\cO(\lambda)$ rounds, where $\lambda$ is the stretch occurring in an execution.
We showed how to design such an algorithm that uses  messages of $\cO(m\log n)$ bits.
Trading off running time for size of messages, we developed an algorithm using smaller messages of only $\cO(n\log n)$ bits that runs in $(\lambda+2)^3$ rounds.
The  running time $(\lambda+2)^3$ of this disconnected-agreement algorithm could possibly be improved asymptotically, while still using messages of $\cO(n\log n)$ bits, but we are sceptical if it is possible to design an early stopping disconnected-agreement algorithm that uses messages of a size that is significantly smaller than $\Theta(m\log n)$ bits. 

We measure the communication efficiency of algorithms by the size of individual messages or the number of non-faulty links used.
This approach allows to demonstrate apparent trade-offs between running time and communication.
One could study dependencies of the running time and the total number of messages exchanged or the total number of bits in messages sent by nodes executing disconnected-agreement algorithms.

\bibliography{disconnected-agreement}

\end{document}